\documentclass{article}

\usepackage{arxiv}

\usepackage[utf8]{inputenc} 
\usepackage[T1]{fontenc}    
\usepackage{hyperref}       
\usepackage{url}            
\usepackage{booktabs}       
\usepackage{amsfonts}       
\usepackage{nicefrac}       
\usepackage{microtype}      
\usepackage{lipsum}
\usepackage{fancyhdr}       
\usepackage{placeins}
\usepackage{amsmath}
\usepackage{caption}
\usepackage{soul}
\usepackage{tikz}
\usepackage{float}
\usepackage{xspace}
\usepackage{xcolor}
\usepackage{amssymb}
\usepackage{graphicx}
\usepackage{longtable}
\usepackage[linesnumbered,ruled,vlined]{algorithm2e}
\usepackage{placeins}

\newcommand{\Real}{\ensuremath{\mathbb{R}}}

\newcommand{\Ratio}{\ensuremath{\mathbb{Q}}}

\floatstyle{ruled}
\newfloat{algorithm}{h}{loa}
\floatname{algorithm}{Algorithm}

\makeatletter
\def\@begintheorem#1#2{\trivlist
	\item[\hskip \labelsep{\bfseries #1\ #2}]\normalfont}
\def\@opargbegintheorem#1#2#3{\trivlist
	\item[\hskip \labelsep{\bfseries #1\ #2\ (#3)}]\normalfont}
\makeatother

\newcommand{\qedsymbol}{$\square$}
\newenvironment{proof}[1][Proof]
{\par\noindent\textbf{#1.}\space}
{\hfill\qedsymbol\par}

\newtheorem{definition}{Definition}
\newtheorem{theorem}{Theorem}
\newtheorem{corollary}{Corollary}
\newtheorem{lemma}{Lemma}
\title{How Pinball Wizards Simulate a Turing Machine}

\author{Rosemary Adejoh \\
  Faculty of Media \\
  Bauhaus University\\
  Weimar, Germany\\
  \texttt{rosemary.utenwojo.adejoh@uni-weimar.de} 
   \And
  Andreas Jakoby \\
  Faculty of Media \\
  Bauhaus University\\
  Weimar, Germany\\
  \texttt{andreas.jakoby@uni-weimar.de} \\
   \And
  Sneha Mohanty\\
  Computer Networks and Telematics\\
  University of Freiburg\\
  Freiburg, Germany\\
  \texttt{mohanty@informatik.uni-freiburg.de} 
   \And
  Christian Schindelhauer\\
 Computer Networks and Telematics\\
 University of Freiburg\\
 Freiburg, Germany\\
  \texttt{schindel@informatik.uni-freiburg.de} \\
}

\begin{document}
\maketitle

\begin{abstract}
We introduce and investigate the computational complexity of a novel physical problem known as the \emph{Pinball Wizard} problem. It involves an idealized pinball moving through a maze composed of one-way gates (outswing doors), plane walls, parabolic walls, moving plane walls, and bumpers that cause acceleration or deceleration. Given the initial position and velocity of the pinball, the task is to decide whether it will hit a specified target point.

By simulating a two-stack pushdown automaton, we show that the problem is Turing-complete---even in two-dimensional space. In our construction, each step of the automaton corresponds to a constant number of reflections. Thus, deciding the \emph{Pinball Wizard} problem is at least as hard as the Halting problem. Furthermore, our construction allows bumpers to be replaced with moving walls. In this case, even a ball moving at constant speed---a so-called ray particle---can be used, demonstrating that the \emph{Ray Particle Tracing} problem is also Turing-complete.
%
%
%
%
%
\end{abstract}


\section{Introduction and Motivation}

\label{intro}
A question of continuous interest is which physical or analog models are Turing complete. Christopher Moore \cite{moore1990unpredictability} outlined a construction showing that, even with precisely known initial conditions, friction less movement, perfect reflection and arbitrary precision, a single billiard ball moving at constant speed can simulate a two-counter machine in 3D. Similarly, Reif et al.~\cite{Reif1994ComputabilityAC} demonstrated in their seminal work that 3D ray tracing can simulate a reversible two-counter automaton, thereby establishing its Turing-completeness. Both models rely on the use of parabolic mirrors in 3D space, supporting Moore's conjecture that three dimensions suffice to achieve Turing universality.

In this paper, we show that full 3D Euclidean space is not required: by allowing the speed of the ball to vary, we introduce a novel two-dimensional model—referred to as the 2D Pinball Wizard Problem, that can simulate a Turing complete two-stack pushdown automaton (PDA)~\cite{hopcroft1979introduction}.

Our {\em Pinball Wizard problem} is inspired by classical pinball machines and features components such as reflecting planes and parabolic walls, one-way gates (outswing doors), bumpers for accelerating and decelerating the ball, as well as a moving plane wall. Unlike in real pinball machines we do not have a tilted plane and the ball, described by a single point, moves in a straight line with constant velocity when it does not bounce into one of the components. These components may change the direction of the ball's trajectory and the velocity of the pinball, when the pinball bounces on these components in its path, see Section~\ref{tracing} for a complete definition of this problem.

The main components of the original pinball game that we retain here, are the one-way gates as well as the target. A one-way gate is an area that is blocked off after the ball passes through it once. It does not change the movement of a ball coming from one side, but will block and reflect the ball if it comes from the other side. We also use bumper walls in our work, which when hit by the pinball, pushes the ball away with a positive or negative change of speed. While there are various types of target in the pinball game, for simplicity, we are working with stationary target position.

The 2D Pinball Wizard model describes a decision problem, where given an initial configuration of the pinball, it is to be decided, whether it hits a given target. We show that this model facilitates the simulation of any two-stack automaton with push and pop operations.

If all stationary bumpers are replaced with a combination of moving walls, the result still holds. We consider the spatial offset interval where the ball enters the construction, and the modified speed determines the time delay (offset) at the end line. These components serve our design of a pinball machine where we want to test if the ball exits at a certain target position given a start time and space offset at the beginning of the path of a pinball having a known speed and direction. Finally, we exchange the ball with a ray particle with constant speed, a modified construction holds up as well, showing that the {\em Ray Particle problem} with moving walls and one-way mirrors is Turing-complete as well.

In Section \ref{rwork}, we discuss Related Work. Thereafter, we provide a more formal definition of the \emph{Pinball Wizard} problem in Section~\ref{tracing}. That Section concludes with a brief mention of the \emph{Ray Particle Tracing} problem, and its relationship with our \emph{Pinball Wizard} problem.  We then demonstrate the Turing-completeness of the \emph{Pinball Wizard} problem by simulating a two-stack pushdown automaton in Section \ref{lowerbound_pinball}, followed by Conclusions and Open Problems in Section \ref{conclusions}.


\section{Related Work}
\label{rwork}

The computational complexity of puzzles and games has long fascinated mathematicians and computer scientists. In his survey \cite{Ryuhei}, Uehara outlines the history of this line of research and shows current trends in theoretical computer science. The simple rules found in puzzles and games not only resemble basic computational operations, but can also serve as alternative models of computation. Studying them may lead to deeper insights into the nature of computation itself.

A remarkable result in this context is the Game of Life, invented by John Conway in 1970 \cite{conway1970game} and popularized by Martin Gardner \cite{gardner1970fantastic}. Confirming a long-standing conjecture, Berlekamp et al. \cite{berlekamp2004winning} showed in 2004 that it is Turing-complete—that is, any computation performed by a Turing machine can be simulated within Conway’s Game of Life. This allows the game to be regarded as a computational model alongside recursive function theory and the Turing machine model.
%

%

Sometimes, games are invented to gain a better understanding of aspects of computation. In 1974, Cook and Reckhow \cite{cook1974storage} introduced the pebble game to study storage requirements. This game has proven to be a powerful tool for exploring computational complexity classes such as NL, P, NP, PSPACE, EXP, and others. It also plays an important role in understanding register allocation, parallel computation, and proof complexity.

There is not enough room to cover all the research on the computational complexity of traditional games and puzzles published since the advent of NP-completeness and higher complexity classes such as PSPACE. For example, Demaine et al.\ in \cite{robotpaper} prove that push-pull block puzzles in 3D are PSPACE-complete to solve, and that push-pull block puzzles in 2D with thin walls are NP-hard, thereby settling an open question posed in \cite{zubaran2011agent}. Push-pull block puzzles are a type of recreational motion planning problem, similar to \emph{Sokoban}, involving the movement of a 'robot' on a square grid with $1 \times 1$ obstacles. While the obstacles cannot be traversed by the robot, some of them can be pushed or pulled into adjacent squares.
Recent progress in the field of combinatorial reconfiguration, particularly through the use of the constraint logic model, has revealed many natural PSPACE-complete problems. These include a variety of sliding block puzzles whose computational complexity had remained open for nearly 40 years.

For example, in \cite{10.1007/978-3-642-13122-6_22}, Forišek analyzes the computational complexity of various two-dimensional platform games. He states and proves several meta-theorems that identify classes of games for which the set of solvable levels is NP-hard, and others for which the set is even PSPACE-hard. Notably, Commander Keen is shown to be NP-hard, while Prince of Persia is shown to be PSPACE-complete. He also investigates the related game Lemmings, constructing instances that admit only exponentially long solutions.

Pinball Wizard is a very popular game, often found in pubs as a contrived mechanical device or as a video/computer game. In the realm of video games, speedrunning is a popular activity where the goal is to complete a game as quickly as possible. Well-known titles such as Super Mario Bros., Castlevania, and Mega Man are played by enthusiasts worldwide, with countless hours spent daily on livestreams as players refine their skills in pursuit of world records. However, human execution is not the only factor in a successful speedrun. Common techniques such as damage boosting and routing require careful planning to maximize time gains. In \cite{lafond:LIPIcs.FUN.2018.27}, Lafond shows that optimizing these mechanics constitutes a significant algorithmic challenge, leading to novel generalizations of classic NP-hard problems such as the knapsack and feedback arc set problems.

The fascination of Pinball Wizard may lie in its physical nature: a ball moving on a plane, representing an analog computational model not captured by traditional digital models discussed so far. Analog models offer many intriguing facets.
In \cite{bonifaci2012physarum}, Bonifaci et al. provide a rigorous proof that a biologically inspired mathematical model of slime mold behavior converges to the shortest path in any network, regardless of topology or initial mass distribution—formally validating classical maze experiments. They generalize the original model by Nakagaki et al. \cite{nakagaki2001smart} to show convergence not only for shortest paths but also for undirected linear programs with non-negative cost vectors.

Siegelmann et al. \cite{siegelmann1991turing} show that a finite neural network composed of sigmoidal neurons can simulate a universal Turing machine, further illustrating the computational power of analog systems.

Analog models can also display surprising limitations. Pour-El and Richards \cite{pour1979computable} construct a computable initial value problem—based on well-behaved differential equations with computable coefficients and initial conditions—whose solution is not computable. This striking result demonstrates that physical systems governed by differential equations may yield unpredictable outcomes, even with perfect knowledge of the inputs.

The reversible nature of physical computation is explored in Fredkin and Toffoli’s Billiard Ball Model (BBM) \cite{fredkin1982conservative}, in which multiple balls and elastic collisions—akin to carom billiards—are used to simulate logic gates. Margolus \cite{margolus1984physics} extends this model to show that conservative, reversible systems can simulate universal computation, making outcome prediction undecidable in the general case. For further constructions, see Durand-Lose’s textbook~\cite{durand2002computing}. As in our work, precise timing is crucial for correctness. Zhang et al.~\cite{zhang2009computing} optimize the mirror-based collision computing model by introducing $m$-counting gates, significantly reducing the number of mirrors required. A comprehensive overview of the field of collision-based computation can be found in \cite{adamatzky2012collision}. In \cite{xorxnor}, Chattopadhyay and Gayen have shown that 3-bit XOR and XNOR logic circuits can be constructed using mechanical elements, such as; beam combiners, fixed as well as movable mirrors.

As mentioned earlier, in 1990 Moore \cite{moore1990unpredictability} envisioned a simulation in which a single billiard ball moving at constant speed in 3D can simulate a two-counter machine. A more detailed proof can be found in the work of Reif et al. \cite{Reif1994ComputabilityAC}, showing that 3D ray tracing (the same problem with a different name) is Turing-complete. Their seminal work analyzes six 3D scenarios and was the first to investigate the computational complexity of ray tracing, establishing it as equivalent to the Halting problem. 
This foundational result has sparked significant theoretical discussions, such as those in \cite{ziegler2009physically}. The intricate nature of Ray Tracing has recently inspired a proposal in \cite{anonymous2025okinawa} to employ optical gates using reflection and refraction as foundational elements in a symmetric encryption algorithm. 

We are not the first to consider the theoretical aspects of bumpers in a pinball machine. Pring et al. \cite{pring2011dynamics} study the dynamics of an impact oscillator with a modified reset law, inspired by a system with "active impact"—namely, the pinball machine. They show how the subtle interplay between two underlying maps gives rise to a large number of new periodic orbits, potentially explaining some of the complexity observed in the dynamics of real pinball machines.

Moore \cite{moore1990unpredictability} and Reif et al. \cite{Reif1994ComputabilityAC} demonstrated that three spatial dimensions are sufficient for Turing-completeness, leaving open the question of whether two dimensions would suffice. While our construction still relies on speed or time as a third dimension, we show that the spatial aspect of the system can be confined to two dimensions while preserving Turing-completeness.




\section{The Pinball Wizard problem}\label{tracing}

We introduce the Pinball Wizard problem as follows: The pinball is a point moving in two dimension Euclidean space with variable speed and direction.  It moves through a pinball machine and interacts with the following components elaborated below under \textit{Components}. 

The ball enters the construction at timestep $t=0$ with a given speed and direction. The task is to determine whether this pinball reaches a given target. We can analyze the movement of the ball within the construction, which consists of stationary and moving walls, one-way gates, parabola walls and bumper walls. With these components we later on construct gadgets implementing two independent stacks using space offset and time delay (offset) while changing the speed on the pinball, when it interacts with our gadgets. Some gadgets, are dedicated to inflict controlled time delays produced by pairwise bumpers or pairwise moving walls. 
\begin{figure}[htb]
	\begin{center}
		\scalebox{.27}{
			\input{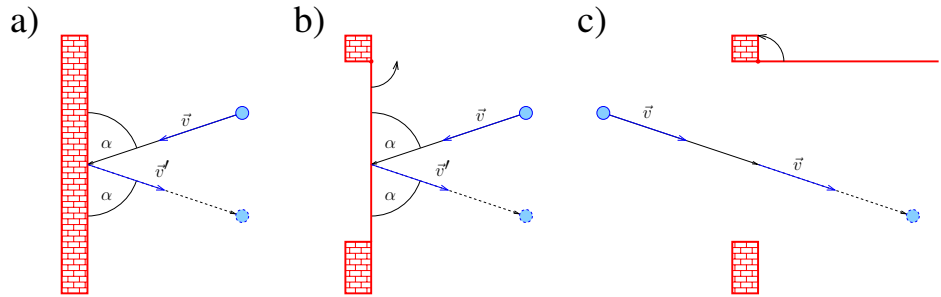}
		}
	\end{center}
	\caption{
		a) A perfectly elastic collision of a ball of speed $\vec{v}$ with a wall, after the collision the direction has changed but the speed is preserved. b) A ball has a collision with a one-way gate on the blocking side. This leads to a perfectly elastic collision. c) A ball has a collision with a one-way gate on the pass through side.
		\label{fig:pinball_tools_1}}
\end{figure}
\FloatBarrier

In the Pinball Wizard problem we always assume perfectly elastic collisions of the ball with walls, where the ball might change speed and direction without any impact on the wall.
%
\paragraph*{Versions of perfectly elastic collisions possible for the Pinball Wizard problem}
\begin{enumerate}
	\item The wall is stationary. Here the ball changes its path after hitting the wall according to the law of reflection
	(see Figure~\ref{fig:pinball_tools_1}.a).
	\item The wall is moving with a constant speed. Then we can investigate the system relative to the moving wall which changes the speed and travel direction of the ball. We determine the new speed and direction of the ball relative to the wall's inertia, see Figure~\ref{fig:pinball_tools_2}.
	\item The wall is accelerated, i.e. the speed changes over time. Then we analyze the speed of the wall at the time of collision, and determine the new direction of the ball and its speed based on the speed and the direction of the movement of the wall at the time of collision. To determine the new direction and speed of the ball, we apply case 2 above.
\end{enumerate}

\begin{figure}[htb]
	\begin{center}
		\scalebox{.26}{
			\input{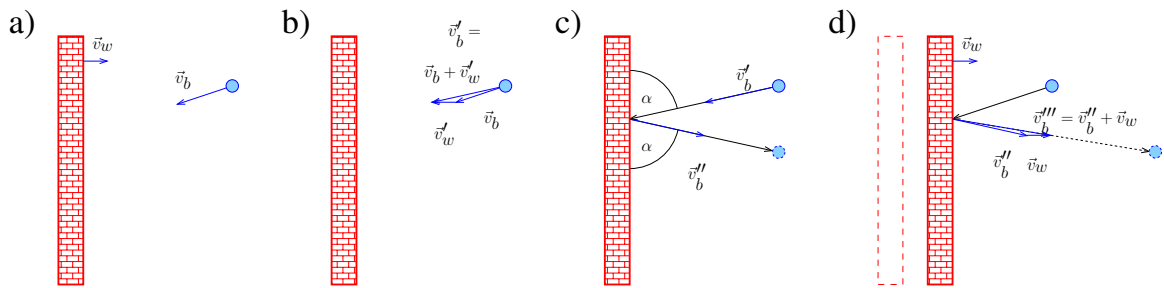}
		}
	\end{center}
	\caption{
		a) A perfectly elastic collision of a ball of speed $\vec{v}_b$ with a wall
		of speed $\vec{v}_w$ b) relative to the wall the ball has a speed of 
		$\vec{v}_b'=\vec{v}_b+\vec{v}_w'$ c) the collision and the 
		speed of the ball $\vec{v}_b''$ relative to the wall after the collision and 
		d) the balls direction and speed $\vec{v}_b'''=\vec{v}_b''+\vec{v}_w$
		if it is not analyzed relative to the moving wall.
		\label{fig:pinball_tools_2}}
\end{figure}
\FloatBarrier

We also use bumper walls, which have a similar effect as straight moving walls. Bumpers increase or decrease the magnitude of the ball’s velocity vector while reversing its direction, reflecting the ball to the opposite direction relative to its path before the collision. This behavior is illustrated in Figure~\ref{fig:pinball_tools_3}.


Another component is the one-way gate. From one side, the ball passes through without any change in its motion; from the other side, it is reflected as if hitting a straight wall—behaving like a trapdoor. This is  
illustrated in Figure~\ref{fig:pinball_tools_1} b) and c). The last type of walls are parabola walls. These static walls are defined by a quadratic equation and the interval of start and end coordinates. The  ball interacts with them by reflection like with a static straight wall, except the curvature has to be taken into account. We consider the tangent of the parabola at the point, where the balls hits this curve. Then, the angle of incidence of the ball is equal to its angle of reflection at this point with respect to the tangent.
Figure~\ref{fig:pinball_tools_4} describes the different types of walls in the problem setup.

\begin{figure}[htb]
	\parbox[b]{0.49\linewidth}{
		\centerline{\scalebox{.25}{
				\input{fig/pinball_tools_3_01.pspdftex}
		}}
		\caption{
			a) A bumper with positive (speed increasing) effect and b) a bumper with negative (speed decreasing) effect. 
			\label{fig:pinball_tools_3}}
	}
	\hfill
	\parbox[b]{0.49\linewidth}{
		\centerline{
			\scalebox{.25}{
				\input{fig/pinball_tools_4_02.pspdftex}
		}}
		\caption{
			a) wall b) moving wall c) one-way gate d) bumper with positive effect e) bumper with negative effect
			\label{fig:pinball_tools_4}}
	}
\end{figure}
\FloatBarrier

\paragraph*{Components}\label{gadget_list}
We now give a formal description of the components of our problem.
\begin{itemize}
	\item {\bf walls} are given by the start and end positions of the wall as rational coordinates in $\Ratio^2$, see Figure.~\ref{fig:pinball_tools_1}.a).
	\item {\bf parabolas} are given by a quadratic function $y=a\cdot x^2+b\cdot x+c$ where $a,b,c$ are rational values as well as start and end values $x_0$ and $x_1$, i.e. $a,b,c,x_0,x_1 \in\Ratio$.
	\item {\bf one-way gates} are given by the start and end positions of the wall given by rational coordinates. We assume that the gate opens in the clockwise direction going from the start position to the end position, see Figure.~\ref{fig:pinball_tools_1}.b) and c).
	\item {\bf moving walls} are given by start and end positions of the wall's initial location as rational coordinates in  $\Ratio^2$,
	the time when movement starts, a time interval $T_0=[0, t_1)$ with 
	a movement function for this time interval described by a rational function in time $t$ given by the polynomials of its numerator and denominator, i.e. the ratio of two polynomials of the time $t$ with rational coefficients, and finally a 
	second time interval $T_1=[t_1, t_2)$ where the wall is moving back. We assume that after step $t_2$ the movement is periodically repeated, see Figure.~\ref{fig:pinball_tools_2}).
	\item {\bf bumpers} are given by start and end positions of the bumper given as rational coordinates, 
	a time when the bumper effect starts, a time interval $T_0=[0, t_1)$ with 
	a bumper effect function for this time interval given by quotient of two polynomials with rational coefficients depending on time $t$, and finally a 
	second time interval $T_1=[t_1, t_2)$ where the bumper rests. We assume that after step $t_2$ the movement is periodically repeated
	(Figure.~\ref{fig:pinball_tools_3}).
\end{itemize}

\begin{definition}[Pinball Wizard problem]
	Given a pinball with rational start and end positions, as well as its speed vector and a given set of components 
	as described above, decide whether the ball reaches the target position.
\end{definition}

	Given that the speed is taken as a constant, for example, the speed of light, then we obtain a similar problem to the Pinball Wizard problem called, the Ray Particle Tracing problem.

\begin{definition}[Ray Particle Tracing problem]
Given a ray particle with rational start and end positions, as well as its speed vector and a given set of components as above, if we assume that walls behave the same as mirrors, decide whether the ray particle reaches the target position. The speed of the ray particle is taken as a constant, which is; the speed of light.
	
\end{definition}

The Ray Particle Tracing problem is closer to the problem analyzed by Moore in \cite{moore1990unpredictability}, as his construction for simulating a two-counter machine in 3D involves a billiard ball moving at constant speed. All proofs and observations for the Ray Particle Tracing problem can be adapted analogously to those for the Pinball Wizard problem. These adaptations are discussed in detail in the Appendix~\ref{sec:RayParticleTracing}.


\section{Pinball Wizard problem is Turing Complete}
\label{lowerbound_pinball}


In this Section, we present the implementation of two independent stacks by using walls, bumper walls, moving walls, and one-way gate for the  
Pinball Wizard problem. For this we use the difference between the 
time and the position when and where the ball arrives at specific lines (i.e. a line called the starting interval or the end interval), more precisely the time and the space offset. In this Section we will present our construction using bumper walls. Note that the bumper walls can be replaced by moving walls in our constructions (which will be discussed in the Appendix~\ref{delay-stack-ball_app}).
 Analyzing this alternative construction we can adapt our result also to the case where the ball never changes its speed, i.e. if it has reached the speed of light. This transforms the Pinball Wizard problem into the Ray Particle Tracing problem, as elaborated in Appendix~\ref{sec:RayParticleTracing}. Hence the Ray Particle Tracing problem is also Turing Complete if we restrict ourselves to 2D plus time.

Within our constructions we have to guarantee that the two stacks are completely independent, i.e. that push and pop operations of one stack do not change the value of the second stack (as a side effect). For the space offset stack this can be guaranteed if all allowed ways of a ball have the same length, thus the time offset will not be effected by the stack operations of the space offset stack. Analogously for the time offset stack we have to guarantee that all allowed ways of the ball are parallel to each other. Thus the operations of time offset stack will not influence the value of the space offset stack. 



\subsection{Implementing a Stack Using the Time Delay of the Pinball}
\label{subsect:Time_Delay_of_the_Pinball}

The main idea to implement a stack by using a time or a space offset works as follows. Assume that the earliest time where the ball can cross a predefined (input) line is $t_0$ and the leftmost position on the line is $s_0$. And assume that the ball crosses the line at time $t_0+\Delta t$ at position $s_0+\Delta s$. Then we can use the offset values $\Delta t$ and $\Delta s$ to implement two stacks. Furthermore, assume a second line is used which represents the output of a gadget implementing a stack operation. Then we use a system of walls, moving walls, bumper walls and one way gates to implement

\begin{equation*}
	\Delta t' \ = \ \left\{\begin{array}[c]{lll}
		\frac{\Delta t}{2} & \text{to implement push 0} & (1)\\
		\frac{\Delta t\ +1}{2} & \text{to implement push 1} & (2)\\
		2\cdot \Delta t & \text{to implement pop, if its output is $0$} & (3)\\
		
		2\cdot \Delta t -1 & \text{to implement pop, if its output is $1$} & (4)
	\end{array}\right.
	\label{eqn:push_pop_equations}
\end{equation*}

%

\begin{figure}[htb]
	\begin{center}
		\scalebox{.34}{
			\input{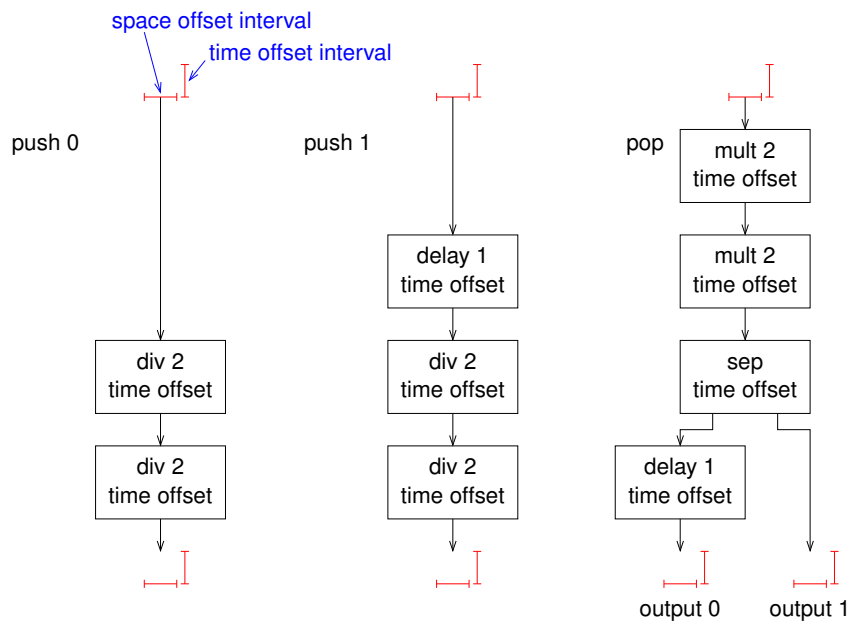}
		}
	\end{center}
	\caption{Illustration of the stack operations implemented by the time offset.
		\label{fig:stack_time_main}
	}
\end{figure}
\FloatBarrier
In the same way we can implement the stack operations using $\Delta s$ for the space offset. For our time offset stack the multiplication with $\frac12$ or with 2 is not sufficient. Since we do not assume infinite speed for the walls, a gap between the two offset intervals is therefore necessary. This gap is implemented by multiplying the time offset two times with $\frac{1}{4}$ or $4$, as required. Thus, we have to multiply with $\frac14$ or $4$ as illustrated in Figure~\ref{fig:stack_time_main}. A detailed explanation is included later in this section. We can draw a similar illustration for the space offset stack (by replacing 'time' with 'space', deleting the second multiplication box, and replacing the extra time offset with a space shift in Figure~\ref{fig:stack_time_main}).

Let us assume that the ball enters the construction with a speed of $v$, then the speed will be modified by hitting the first bumper wall. 
By using a second bumper wall we will take the speed modification back. Hence the ball leaves the construction with a speed of $v$.
The multiplication factor of the time offset will be generated by the modified speed within the time where the ball has a modified speed 
and the differed distances of the bumper walls. Recall, if the first bumper wall has a positive effect to the speed of the ball, the second has to have a 
negative effect and vice versa. We assume that the effect can by modified over time, analogously to acceleration. Therefore we use the acceleration
of a bumper to describe this behavior of a bumper.

So the bumpers are used to manipulate the 
speed of the ball, and the changed speed over a given fixed distance $\delta_2$ 
is used to achieve a multiplication of $2$ (or $\frac{1}{2}$, resp.) with the time offset.
We assume that outside of the construction blocks, the ball maintains a given speed $v$. This allows us to keep track of the time interval in which the ball arrives at a specific construction block.

The construction is illustrated in Figure~\ref{fig:tool_time_bumper_mult_1}.
We now claim that:
	%
	%

\begin{lemma}
	The 2D Pinball Wizard problem with one-way gates, plane and parabolic walls, moving walls, and bumpers
	simulates a stack using the time offset of the ball.
	\label{th:Pinball_Bumper}
\end{lemma}

\begin{proof}
	We show the proof by describing how the multiplication of the time offset by 2 is implemented.
	Let us assume the following values
	where $t_s$ is the starting time of the ball relative to a base time line.
	We can divide the way of the ball into three parts. Each of these parts have been described in Figures~\ref{fig:tool_time_bumper_mult_1} (way 1, 2 and 3).
	
	\noindent
	{\bf Part 1}
	In the first part the ball goes from the 
	starting line to the first moving wall.
	\begin{itemize}
		\item $\delta_1$ is the distance between the starting line of the ball and the 
		start line of the moving interval of the first moving wall.
		\item $v$ is the initial speed of ball and its final speed.
		\item The ball reaches first the bumper wall at time step $t_1(t_s)$.
		Note that $t_1(t_s) = t_s+\frac{\delta_1}{v}$. We use $t_1=t_1(0)$.
		\item $a_1(t)$ is the acceleration of the first bumper wall. We assume that
		the acceleration is 0 at time $t\le t_1$. The effect of the bumper wall will be $v_1(t)= (t-t_1)\cdot a_1(t)$.
		\item $v_{b,1}(t_s)$ is the speed of the ball after the collision with the first bumper wall.
		Note that $v_{b,1}(t_s) = v - v_1(t_1+t_s)$.
	\end{itemize}
	Note that 
	$$
		t_1  =  \frac{\delta_1}{v}\qquad\text{and}\qquad
		v_1(0)  =  v\ .
	$$
	
	\begin{figure}[htb]
		\begin{center}
			\scalebox{.25}{
				\input{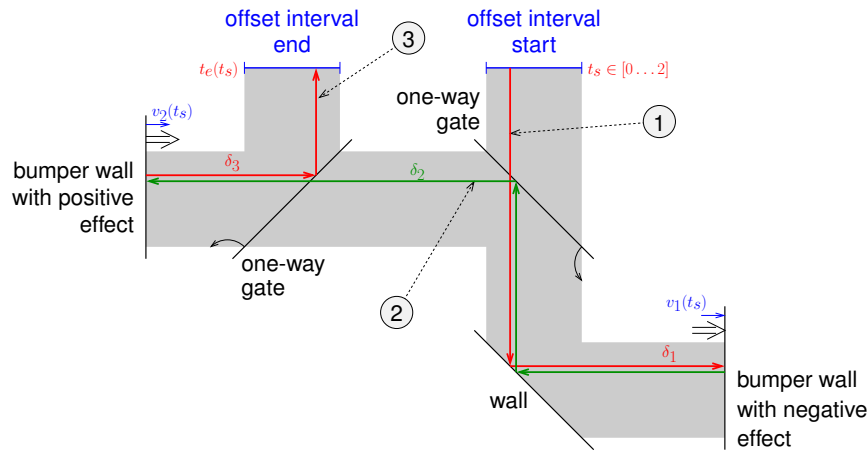}
			}
		\end{center}
		\caption{Multiplying the time offset by a factor using bumpers: The first red way (way 1) from the input (at offset interval start) to the right hand side bumper wall illustrates the the length $\delta_1$. The green way (way 2) between the two bumper walls illustrates the the length $\delta_2$. This is the only way where the ball has a changed speed. The second red way (way 3) from the left hand side bumper wall to the output (at offset interval end) illustrates the the length $\delta_3$. The way between the two bumper walls is used to implement the multiplication of the time offset by a constant. The vertical shift of the 3 ways is used to increase readability. \label{fig:tool_time_bumper_mult_1}}
	\end{figure}
	
	\noindent
	\newpage
	{\bf Part 2}
	In the second part the ball goes from the 
	the first bumper wall to the second bumper wall.
	\begin{itemize}
		\item $\delta_2$ is the distance between the 
		first bumper wall and the second bumper wall.
		\item $v_{b,1}(t_s)$ is the speed of ball within this part.
		\item The ball reaches the second bumper wall at time step $t_2(t_s)$.
		Note that $t_2(t_s) = t_1(t_s)+\frac{\delta_2}{v_{b,1}(t_s)}=t_s+\frac{\delta_1}{v}+\frac{\delta_2}{v_{b,1}(t_s)}$. We use $t_2=t_2(0)$.
		\item $a_2(t)$ is the acceleration of the second bumper wall. We assume that
		the acceleration is 0 at time $t\le t_2$. The effect of the bumper wall will be $v_2(t)= (t-t_2)\cdot a_2(t)$.
		\item $v_{b,2}(t_s)=v$ is the speed of the ball after the collision with the second bumper wall.
		Note that $v_{b,2}(t_s) = v_{b,1}(t_s)+v_2(t_2(t_s))$.
	\end{itemize}
	Note that 
	$$
		t_2  =  \frac{\delta_1+\delta_2}{v}\qquad\text{and}\qquad
		v_{b,2}(0)  =  v_{b,1}(0) \ = \ v\ .
	$$
	
	\noindent
	{\bf Part 3}
	In the third part the ball goes from the 
	the second bumper wall to the end line.
	\begin{itemize}
		\item $\delta_3$ is the distance between the 
		second bumper wall and the
		end line of the ball.
		\item The ball reaches the end line at time step $t_c(t_s)$.
		Note that $t_e(t_s) = t_2(t_s)+\frac{\delta_3}{v}$. We use $t_e=t_e(0)$.
	\end{itemize}
	Note that
	$$
		t_e  =  \frac{\delta_1+\delta_2+\delta_3}{v}\qquad\text{and}\qquad
		t_e(t_s)  =  t_s+\frac{\delta_1+\delta_3}{v}+\frac{\delta_2}{v_{b,1}(t_s)}\ .
	$$
	Our goal is now to determine functions for the acceleration $a_1(t_s)$ and $a_2(t_s)$ 
	such that $t_e(t_s) = 2\cdot t_s + c$ for a value $c=t_e$
	which does not depend on $t_s$. This leads to \\
	
	and therefore \hspace*{26mm}
	$
	\begin{array}[b]{rcl}
		\displaystyle t_s + \frac{\delta_2}{v} & = & 
		\displaystyle \frac{\delta_2}{v_{b,1}(t_s)} \\[3mm]
		\displaystyle v_{b,1}(t_s) & = & 
		\displaystyle \frac{v\cdot \delta_2}{v\cdot t_s + \delta_2}\ .
	\end{array}
	$
	
	Hence
	\begin{eqnarray*}
		a_1(t_1+t_s) & = & \frac{v-v_{b,1}(t_s)}{t_s} \ \ = \ \ 
		\frac{v^2}{v\cdot t_s + \delta_2}\\
		a_2(t_2+2\cdot t_s) & = & \frac{v-v_{b,1}(t_s)}{2\cdot t_s} \ \ = \ \ \frac{v^2}{2\cdot(v\cdot t_s + \delta_2)}
	\end{eqnarray*}
	
	
\end{proof}

For the multiplication of the time offset by $\frac{1}{2}$ we exchange the two bumpers and the analysis will follow the same ideas as above.

\begin{figure}[htb]
	\begin{center}
		\scalebox{.35}{
			\input{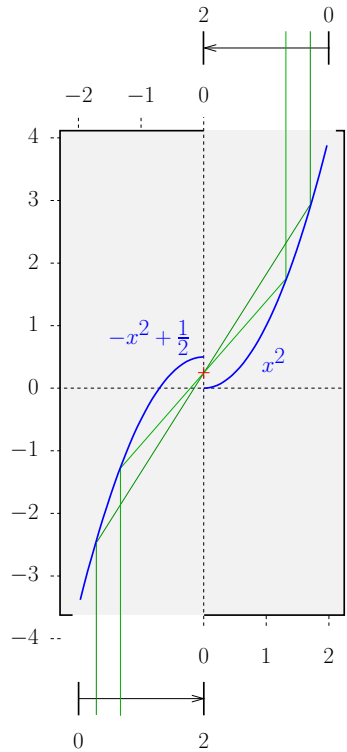}
		}
	\end{center}
	\caption{After performing a multiplication step, the order of the offset is reversed (the zero offset is moved from the left of the pinball interval to the right of the interval). This tool moves the pinball to its original order, i.e. the zero offset is moved back to the left side of the interval. Depending on the widths of the pinball interval, one might have to implement a divergent version of this tool, i.e. with $x$ coordinates going from $-4$ to $4$ instead of the presented version with $x$ coordinates going from $-2$ to $2$. Within the following construction only offsets between $0$ and $2$ (excluding $0$ and $2$) are used.
		\label{fig:revOffset_appendix}}
\end{figure}

Note that at the output interval the order of the space offset is reversed. To correct the order we have to add the construction of Figure~\ref{fig:revOffset_appendix} at the output interval of the time offset multiplier.

It is also possible to implement a multiplier for the time offset by using only moving walls (see Figure~\ref{fig:tool_time_no_bumper_mult_1}). The analysis is shown in the Appendix~\ref{delay-stack-ball_app}. 
The implementation of the multiplier by moving walls can be adapted such that it can be used to handle the case where the ball does not change its speed (leading to the Ray Particle Tracing Problem). 
If 
we simulate bumpers with moving walls the movement of both walls has to be coordinated in such a way that the different speeds and delays add up to the same effect as two bumpers. 

Finally we have to describe a gadget for adding a constant delay to the ball and a gadget for separating early balls (arriving a specific position with a time offset of $[0\ldots 1)$) from late balls (arriving at a specific position with a time offset of $[1\ldots 2)$). The constant delay can easily be generated by using some $45^\circ$ walls. Hence we will focus now on the construction of the time offset separator.

The main idea of the time offset separator is to use a moving wall which leads a ball arriving with a time offset $[0\ldots 1)$ in one distinct direction and a ball arriving with a time offset $[1\ldots 2)$ in a second distinct direction. Since we do not assume that a wall can move with an infinite speed, we have to implement a gap between the two offset intervals. This can be implemented be doubling the multipliers, i.e. instead of multiplying the time offset once with $\frac12$ (or with $2$, respectively), we multiply the time offset two times with $\frac12$ (or with $2$, respectively). Hence we multiply the time offset either with $\frac14$ or with $4$. 

As illustrated in Figure~\ref{fig:stack_time_main} the time offset separator only occurs after the multiplication with $4$ within the pop operation. Thus we can assume that the ball will arrive with a time offset within the interval $[0\ldots 2)$ (instead of within the standard interval $[0\ldots 1)$). Using double multiplication guarantees that the ball will either arrive with a time offset in $[0\ldots \frac12)$ or with a 
time offset in $[1\ldots \frac32)$.


%
\begin{figure}[htb]
	\begin{center}
		\scalebox{.25}{
			\input{fig/tool_time_mult_04_nn.pspdftex}
		}
	\end{center}
	\caption{As in the case of the multiplication by using bumpers, we can split the way from the starting interval to the end interval into 3 part: The first red way (way 1) from the input (at offset interval start) to the right hand side moving wall, the green way (way 2) between the two moving walls, and the second red way (way 3) from the left hand side moving wall to the output (at offset interval end). Way 2 
		is used to implement the multiplication of the time offset by a constant. 
		\label{fig:tool_time_no_bumper_mult_1}}

	\begin{center}
		\scalebox{.17}{
			\input{fig/tool_if_else_time_01_pb.pspdftex}
		}
	\end{center}
	\caption{Construction for separating the ball arriving with a time offset in the time interval $[0\ldots\frac12)$ 
		from the ball arriving with a time offset in the time interval  $[1\ldots\frac32)$. Then using the distance of $d-1$ instead of
		$d$ shifts the base of the time interval of the second interval $[1\ldots\frac32)$ from $1$ to $0$. The time gaps in the offset between $\frac12$ to $1$ (and from $\frac32$ to $2$) are necessary to give the moving walls the time either to move out of the way of the ball or to move into its initial position.
		\label{fig:if_else_time_pb}
	}
\end{figure}
\FloatBarrier

This allows the separating 
moving wall to move out of the way and to move back. The construction is illustrated in Figure~\ref{fig:if_else_time_pb} and is analyzed in Appendix~\ref{delay-stack-ball_app}.

%



\subsection{Implementing a Stack Using the Space Offset of the Pinball}
\label{space_offset_stack}

The space offset stack can be implemented using a system of plane walls, one-way gates and parabola walls. Note that the distance the ball has to pass to reach the focus is independent of its starting position on the starting interval. This distance is given by the distance from the starting line segment to the point where the ball hits the parabola wall plus the distance between this point to the focus.

Let us assume that a ball going down in the direction of the $x$-axis starts at an arbitrary point $\ell$ on the line segment defined from point $(0,h)$ to point $(2,h)$ at time $t$. Assume again that the ball bounces off a parabola $(x, a\cdot x^2)$ with focus $(0, \frac{1}{4a})$. If the ball starts at point $(\ell,h)$ then the distance traveled by the ball to the focus is given by 
$$
h - a\cdot \ell^2 \quad + \quad a\cdot \ell^2 + \frac{1}{4a} \qquad = \qquad h+\frac{1}{4a}
$$
which is independent of the starting point at the starting line segment (see Figure~\ref{fig:tool_parabola_mirror_02_rpt}). If we assume that 
within each time unit the ball can pass through one space unit, then within a construction like in Figure~\ref{fig:mult05_01_rpt}, it will reach the focus at time $t+h+\frac{1}{4a}$, independent of $\ell$.

\begin{figure}[htb]
	\centering
	\begin{minipage}[t]{0.3\textwidth}
		\centerline{
			\scalebox{.23}{
				\input{fig/tool_parabola_mirror_02.pspdftex}
		}}
		\caption{A parabola with points $(x, a\cdot x^2)$ and focus $(0, \frac{1}{4a})$.
			\label{fig:tool_parabola_mirror_02_rpt}}
	\end{minipage} 
	\hspace*{0.05\textwidth}
	\begin{minipage}[t]{0.53\textwidth}
		\centerline{
			\scalebox{.22}{
				\input{fig/tool_mult_0.5_01.pspdftex}
		}}
		\caption{Multiplication of the offset by $\frac{1}{2}$ is a main tool for implementing a push operation. Here we show the reverse ordering
			when two separate balls enter the starting interval line at different positions.	\label{fig:mult05_01_rpt}}
	\end{minipage} 
\end{figure}

\begin{lemma}
	The Pinball Wizard problem can simulate a stack which is implemented by the space offset of the ball within a system of parabola walls.
	\label{le:pinball_space_offset_stack}
\end{lemma}

\begin{proof}
	We achieve push and pop operations using the multiplication of the offset by $2^{-1}$ or by $2$. We assume starting interval between 0 and 2.	For the multiplication of the offset by $\frac{1}{2}$, two parabolas are placed such that the second parabola is flipped as seen in Figure~\ref{fig:mult05_01_rpt}. Analogously, multiplying the offset by $2$ can be done using a similar construction (see e.g. Figure~\ref{fig:mult2_01_appendix} in Appendix~\ref{space_stack_app}). As with the time offset stack, both constructions reverse the ordering of space offset of the ball. Thus a similar construction as in Figure~\ref{fig:revOffset_appendix} can be used to correct the ordering. 
	
	Larger offsets are allowed to deal with pushing 1 to or popping 1 from the stack, as both operations require the addition of a 1 to the offset (that is, a shift in the offset interval).

	\begin{figure}[htb]
		\centering
		\begin{minipage}[t]{0.45\textwidth}
			\centerline{
				\scalebox{.15}{
					\input{fig/tool_push_0_01.pspdftex}}
			}
			\caption{Multiplying the offset from the interval $[0,1)$ by $\frac{1}{2}$ implements a push of the Boolean value
				0 to the stack. The yellow area represents area of possible offsets.
				\label{fig:push0offse}}
		\end{minipage} 
		\hspace*{0.05\textwidth}
		\begin{minipage}[t]{0.45\textwidth}
			\centerline{
				\scalebox{.15}{
					\input{fig/tool_push_1_01.pspdftex}
			}}
			\caption{Shifting the offset from the interval $[0,1)$ to the interval $[1,2)$ and 
				multiplying the resulting offset by $\frac{1}{2}$ implements a push of the Boolean value
				1 to the stack. The yellow area represents area of possible offsets.
				\label{fig:push1offset}}
		\end{minipage} 
	\end{figure}
	
	For push 0 we multiply the space offset by $\frac12$ (see Figure~\ref{fig:push0offse}).
	To push 1 we first shift the starting offset interval from $[0\ldots 1)$ to $[1\ldots 2)$ before multiplying the space offset by $\frac12$ (see Figure~\ref{fig:push1offset}), and for pop 1, after multiplying by $2$, the possible output offset values are split in two intervals, each going from 0 to 1 as seen in Figure~\ref{fig:if_else_time_space}. Thus if the top entry of the stack is a 0, the ball with the modified offset will show up on the left exit. In the case where it is a 1, the ball will accrue to the right exit. To separate these two cases we have to split the interval after the multiplication into two intervals to get the resulting stack and the output of the pop operation. 
	
	\begin{figure}[htb]
		\begin{center}
			\scalebox{.13}{
				\input{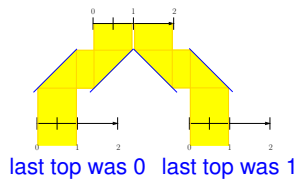}
			}
		\end{center}
		\caption{Construction for separating the arrival of the ball within the space 
			interval $(0\ldots 1)$ from the ball's arrival within the space interval $(1\ldots 2)$.
			\label{fig:if_else_time_space}
		}
	\end{figure}
	
	By assuming the bottom of the stack always contains at least two 1 digits, the offset can never have the values $0,\frac{1}{2},$ and $1$. To guarantee that these edge cases never occur we can assume a bottom symbol on the stacks which will never be removed. 
\end{proof}

The full proof as well as constructions associated with Lemma~\ref{le:pinball_space_offset_stack} are given in Appendix~\ref{space_stack_app}.

One should highlight here that from Figure~\ref{fig:tool_parabola_mirror_02_rpt}, it can be seen that for every offset of the ball that might occur within each of the push 0, push 1 and pop constructions, the ball will always need the same amount time to go from the horizontal start line to the finish line. This is possible because an extra delay can be introduced on the ball's movement by an additional substructure such that its time of arrival at the end line is well-defined. Hence, the operations of the space offset stack is independent of the time offset stack.

Note that both pop operations (for the time offset stack as well as for the space offset stack) results in two potential output intervals (one for the pop result 0 and one for the pop result 1). When we would like to continue with our simulation of the two-stack PDA we have to combine the two intervals again. For this we can use a one-way gate, which guarantees a perfect overlapping of the two intervals (the 0 offset of the first interval will overlap with the 0 offset of the second interval and the 1 offset of the first interval will overlap with the 1 offset of the second interval).

To illustrate the behavior of space as well as time offset intervals w.r.t. push and pop operations we will showcase a short example in the following :

\begin{figure}[htb]
	\begin{center}
		\scalebox{.25}{
			\input{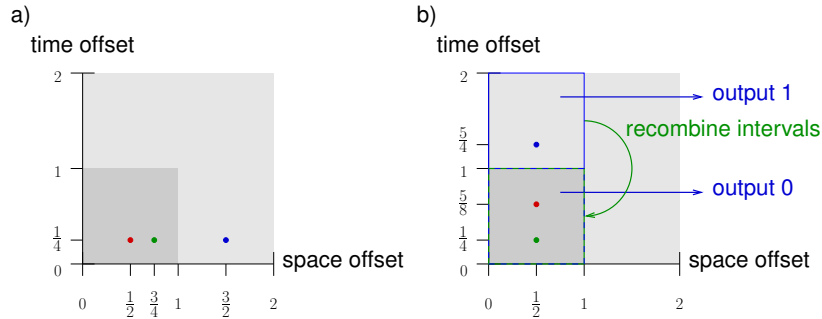}
		}
	\end{center}
	\caption{Behaviour of the time and space offset intervals according to the (a) push(1) and (b) pop operations. The red dot denotes the initial position of the stack, the blue dot shows the intermediate position and the green dot shows the final position of the stack
		\label{fig:offsets_01_main}}
\end{figure}
\FloatBarrier

For the behaviour of the offset intervals according to stack operations (see Figure~\ref{fig:offsets_01_main} a), starting with the stack values $01$ (represented by the binary value $0.01_2$) for the time offset stack and $1$ (represented by the binary value $0.1_2$) for the space offset stack (red dot), we begin by performing a push$(1)$ operation to the space offset stack: First add 1 to the value of the space offset stack (blue dot), and then multiply by $2^{-1}$ (green dot). This leads to the stack content $11$ (represented by the binary value $0.11_2$).

In Figure~\ref{fig:offsets_01_main} b), we start with stack values $101$ (represented by the binary value $0.101_2$) for the time offset stack and $1$ (represented by the binary value $0.1_2$) for the space offset stack (red dot) and perform a pop operation on the time offset stack as follows: First multiply the value of the time offset stack by $2$ (blue dot). If the result is in the range between 1 and 2, the output will be 1, otherwise if the result is in the range between 0 and 1, the output will be 0. Finally we recombine the two possible intervals (by moving the upper interval to map on the lower interval, resulting in the green dot). This leads to the stack content $01$ (represented by the binary value $0.01_2$) in the time offset stack. Note that all the operations should not modify the space offset value.

One should highlight here that from Figure~\ref{fig:tool_parabola_mirror_02_rpt}, one can see that for every offset of the ball that might occur within each of the push 0, push 1 and pop constructions, the ball will always need the same amount of time to go from the horizontal start line to the finish line. This is possible because an extra delay can be introduced on the ball's movement by an additional substructure such that its time of arrival at the end line is well-defined. Hence, the operations of the space offset stack is independent of the time offset stack.\\

Now, using the proofs of Lemma~\ref{th:Pinball_Bumper}, Lemma~\ref{le:pinball_space_offset_stack} as well as Lemma~\ref{le:Pinball_noBumper} (discussed in Appendix~\ref{delay-stack-ball_app}), we conclude that -

\begin{theorem}
	The 2D Pinball Wizard problem with one-way gates, plane, parabolic walls and moving walls
	simulate a two-stack pushdown automaton (PDA) such that each step is represented by a constant number of reflections. This also holds if we replace the moving walls by bumpers. 
	\label{th:Pinball_noBumper}
\end{theorem}
Since a two-stack PDA can simulate any deterministic Turing machine, we conclude the following:

\begin{corollary}
	Deciding the 2D Pinball Wizard problem with one-way gates, plane, parabolic and moving walls is as hard as solving the Halting problem.
\end{corollary}

\newpage


\section{Conclusions and Open Problems}
\label{conclusions}

Inspired by pinball machines, the Pinball Wizard problem involves navigating a maze composed of stationary and moving walls, one-way gates, parabola walls and bumper walls. The objective is to determine whether a pinball can reach a designated target.
 We show that the Pinball Wizard problem in two dimensions can simulate a two-stack PDA, which is equivalent in power to a single-tape deterministic Turing machine. Thus, solving instances of this problem is at least as hard as deciding the Halting problem.

While the billiard problem with static walls—also known as ray tracing with mirrors—is known to be Turing-complete in three dimensions, as established independently by Moore~\cite{moore1990unpredictability} and Reif et al.~\cite{Reif1994ComputabilityAC}, our work may contribute to resolving the open question of whether Turing-completeness can also be achieved in two dimensions.

Another interesting direction for future research is to investigate whether a simplified model with a reduced set of components can still achieve Turing-completeness. Currently, our construction relies on idealized assumptions such as arbitrary precision and perfect reflection. It remains an open question how the computational complexity would change under more realistic physical constraints, including damping, friction, finite ball diameters and random disturbances. Future work could also explore whether a smaller set of components suffices for the Pinball Wizard problem specifically and what precision bounds are necessary for reliable computation.

\bibliographystyle{unsrt}  
\bibliography{lit_for_illumination-mirror_tools_for_time}

\newpage

\newpage
\appendix

\section{Appendix}\label{sec:lemmaappendix}
In this Appendix we present all technical details and full proofs.
We start with presenting more details and the analysis of the stack implementation using space and time offsets for the Pinball Wizard problem. Thereafter we present the Ray Particle Tracing problem. The simulation of a Turing Machine using both stacks follows standard techniques therefore we skip a full discussion on it.



\section{Implementing a Stack Using the Space Offset of the Pinball}
\label{space_stack_app}

Recall that to simulate a Turing machine we would like to implement two independent stacks by two different properties of the ball. Within the implementation of each stack, we have to guarantee that no operation on the first stack (i.e. pop or push) influences the characteristics of the second stack and vice versa. 
For the implementation of the first stack, we have used the time offset when the ball reaches a specific line interval. For the second stack
we will use the space offset of a ball according to a base position when crossing this line interval. 

For implementing pop and push operations, we use systems of 
walls, one way gates, and parabola walls. 

Before going into the details it should be mentioned that
for every parabolic wall given by the formula $a\cdot x^2+b$ with $a,b\in\Real$ and
$a>0$, every ball coming from above the wall and going 
parallel to the $y$-axis which hits the wall has to pass through the 
focus of the parabola given by $(0,\frac{1}{4\cdot a}+b)$, 
see e.g. Figure~\ref{fig:focusPara_appendix}.

\begin{figure}[htb]
\begin{center}
\scalebox{.45}{
\input{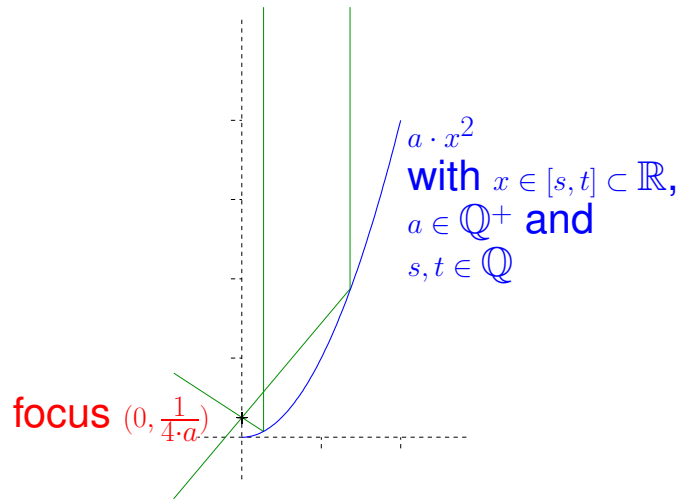}
}
\end{center}
\caption{Focus of a parabolic wall.\label{fig:focusPara_appendix}}
\end{figure}

Recall that this stack will be implemented by the (horizontal) space offset and the other stack by the time offset. Hence it is necessary, that the 
space offset has no influence on the time that is needed for a ball to pass through each of the used parts of the construction of the space offset stack. 
For this let us assume that a ball going down in the direction of the $y$-axis starts at an arbitrary point of the line segment 
from point $(0,h)$ to point $(2,h)$ at time $t$. Let us assume that the ball is reflected at a parabola $(x, a\cdot x^2)$ with focus 
$(0, \frac{1}{4a})$. Note that the distance the ball has to pass to reach the focus is given by the distance from the starting line segment to the 
point where the ball hits the parabola plus the distance between this point and the focus. 
In Subsection~\ref{space_offset_stack} we have already seen that this distance (and therefore the traveling time of the ball) is independent of the concrete starting position within the line segment 
from point $(0,h)$ to point $(2,h)$.



We assume that a ball without an offset passes a specific horizontal line at a 
specific position. The (horizontal) distance between this position and the 
actual position where the ball passes the line will be called the offset.
For the offset of the ball we assume in the following, 
values between 0 and 1 (excluding the extreme values). Given an offset of $s\in (0\ldots 1)$ we can implement a push of a bit 
$b\in\{0,1\}$ by 
$$
s' = (s+b)\cdot 2^{-1}
$$
and a pop by 
$$
s' = (2\cdot s) - \lfloor 2\cdot s\rfloor\qquad\text{and determining}\qquad 
b = \lfloor 2\cdot s\rfloor
$$
One can see that the multiplication by $2$ or by $2^{-1}$ is a major tool for these operations. To implement this we use a combination of parabola walls. \\
For the multiplication with $2^{-1}$ (see Figure~\ref{fig:mult05_01_rpt})
$$
- 2\cdot x^2+\frac{3}{8} \quad\text{ for } x\in [-1, 0]\quad\text{ and } \qquad x^2
\quad\text{ for } x\in [0, 2]\ .
$$
For the multiplication with $2$ (see Figure~\ref{fig:mult2_01_appendix})
$$
- \frac{1}{2}\cdot x^2+\frac{3}{4} \quad\text{ for } x\in [-4, 0]\quad\text{ and } \qquad x^2
\quad\text{ for } x\in [0, 2]\ .
$$
In our figures we have used the possibility of having larger offsets. This 
is based on the requirement that we have to add a $1$ to the offset in some 
scenarios, especially if we would like to push a $1$ to our stack, or we have 
to pop a $1$ from the stack. In these cases the construction takes care of 
this larger offset, and either divides the value of the offset by 2 or splits the 
interval of possible offset values into two\textemdash with each interval going from $0$ to $1$.

\begin{figure}[H]
\begin{center}
\scalebox{.34}{
\input{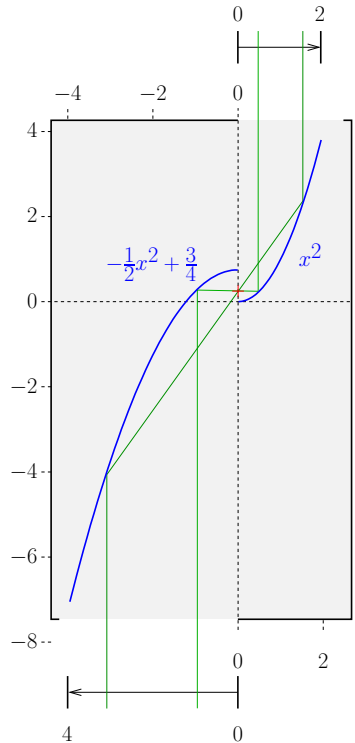}
}
\end{center}
\caption{Multiply the offset by 2 and reverse the ordering of the possible offsets of the ball. We use the binary digits of the decimal place of the offset as a representation of a stack of boolean values. Hence multiplying the offset by 2 is a main tool for implementing a pop operation.\label{fig:mult2_01_appendix}}
\end{figure}

Note that after using one of these constructions the offset interval is reversed. 
If on the input the offset was assumed to move the ball to the right of the $0$-line, it should also move the ball to the right of the $0$-line at the output and not to the left. To correct this 
we can use two extra parabolic walls (see Figure~\ref{fig:revOffset_appendix})
$$
- x^2+\frac{1}{2} \quad\text{ for } x\in [-2, 0]\quad\text{ and } \qquad x^2
\quad\text{ for } x\in [0, 2]\ .
$$


Using these modules we can now implement the required stack operations for the space offset stack, 
see Figure~\ref{fig:push0offset_appendix}, \ref{fig:push1offset_appendix}, and \ref{fig:popoffset_appendix}.
The equations associated with each operation can be found in Section~\ref{eqn:push_pop_equations}.
For push 0 we simply divide the offset by 2 (see (1)). 
For push 1 we first shift the space offset interval of the ball one space unit to the right, implement an addition of 1 and 
afterwards divide the offset by 2 (see (2)). The pop module has two possible 
output intervals for the offset; the first is used if the pop 
results in a $0$ and the second if the pop results in a $1$. To get these two different 
intervals we multiply the offset first with 2 (see (3) and (4)).

\begin{figure}[htb]
\centering
\begin{minipage}[t]{0.45\textwidth}
\centerline{
		\scalebox{.22}{
			\input{fig/tool_push_0_01.pspdftex}}
		}
	\caption{Multiplying the offset from the interval $[0,1)$ by $\frac{1}{2}$ implements a push of the Boolean value
0 to the stack. The yellow area represents area of possible offsets.
		\label{fig:push0offset_appendix}}
\end{minipage} 
\hspace*{0.05\textwidth}
\begin{minipage}[t]{0.45\textwidth}
\centerline{
		\scalebox{.22}{
			\input{fig/tool_push_1_01.pspdftex}
		}}
	\caption{Shifting the offset from the interval $[0,1)$ to the interval $[1,2)$ and 
multiplying the resulting offset by $\frac{1}{2}$ implements a push of the Boolean value
1 to the stack. The yellow area represents area of possible offsets.
		\label{fig:push1offset_appendix}}
\end{minipage} 
\end{figure}

\begin{figure}[htb]
\begin{center}
\scalebox{.14}{
\input{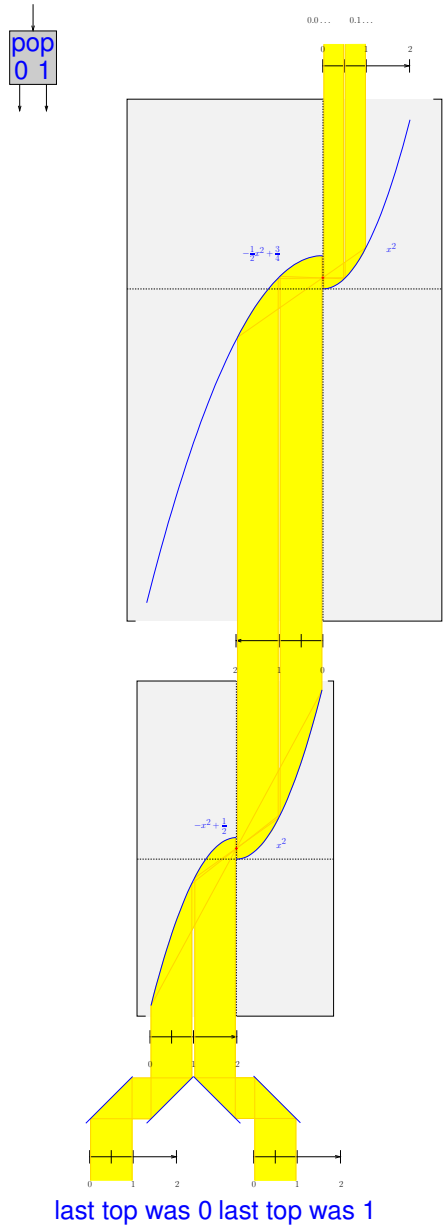}
}
\end{center}
\caption{Multiplying the offset by 2 is a basic part for implementing the pop operation.
If the top entry of the stack is 0 then the ball with modified offset (the top element is removed) will accrue on the left exit. In the case it is 1 the ball with modified offset will accrue on the right exit.
If we assume that the bottom of the stack always contain at least two 1 digits, then our offset can never assume the values $0,\frac{1}{2},$ and $1$. The yellow area represents area of possible offsets.\label{fig:popoffset_appendix}}
\end{figure}
\FloatBarrier

Based on the observation in Figure~\ref{fig:tool_parabola_mirror_02_rpt}
one can see that for every offset of a ball that might occur within every single construction of this Section, the ball will always need the same amount of time to go from the horizontal starting line above the construction 
to a horizontal finishing line below the construction. 

To implement a stack we only need to use the construction for implementing a push 0, a push 1, and a pop operation (Figure~\ref{fig:push0offset_appendix},~\ref{fig:push1offset_appendix},~\ref{fig:popoffset_appendix}). 
Note that the time for the three different constructions might differ. Using the construction given in 
Figure~\ref{fig:tool_add_delay}, we can add to any ball some delay such that for every offset within 
any of the three constructions, a ball requires the same time to reach a 
predefined horizontal finishing line from a predefined horizontal starting line.

\begin{figure}[htb]
\begin{center}
\scalebox{.25}{
\input{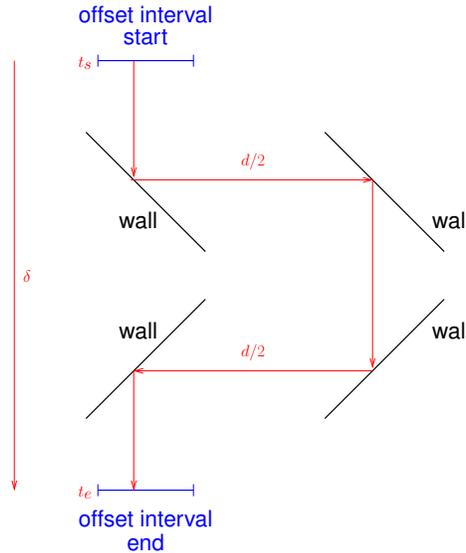}
}
\end{center}
\caption{Add the delay of $d$ to a ball. Note that $d$ has to be large enough such that there is no 
overlap of the two upper walls or lower walls.\label{fig:tool_add_delay}
}
\end{figure}
\FloatBarrier

Note that the different sub-constructions are connected by some corridors which act like optical fibres.
Based on the length of such a corridor connection, a ball might have a connection-specific delay. 
By using the construction of Figure~\ref{fig:tool_add_delay} we can guarantee that any ball arrives at
any sub-construction at a well defined time slot.

The main idea of implementing a TM (or a PDA or a FA) which might use a binary stack instead of a tape is that we use the 
two different horizontal finishing lines of the pop construction to implement the required behaviours of the
TM after getting a 0 or a 1 from the stack. After the constructions which determine the different 
transitions of the TM (or a PDA or a FA)\textemdash and may be after pushing some values to the stack\textemdash we might have to rejoin the 
two different space offset intervals in such a way that the required time for any ball to go 
from both input offset intervals to output interval is constant.
This can be implemented by the construction of Figure~\ref{fig:tool_rejoin_intervals_01}.

\begin{figure}[htb]
		\centerline{
			\scalebox{.20}{
				\input{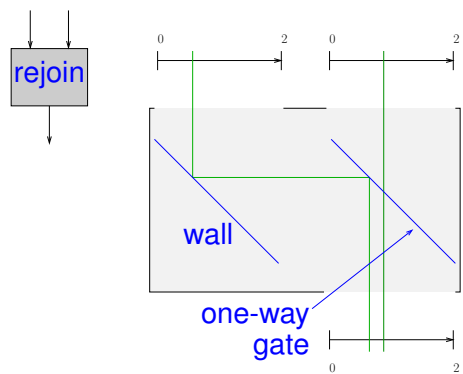}
			}
		}
		\caption{A simple rejoin tool to be used after the splitting at a pop operation.
			\label{fig:tool_rejoin_intervals_01}}
\end{figure}
\FloatBarrier





\section{Using the Delay of a Pinball without Bumpers to Implement a Stack}
\label{delay-stack-ball_app}

In addition to bumpers, moving walls can also be used to change the speed of a ball. This change in the speed can result in a multiplication of the time offset by a constant factor. 
The construction of this multiplier is illustrated in Figure~\ref{fig:tool_time_ball_mult_1}, \ref{fig:tool_time_ball_mult_2}, and~\ref{fig:tool_time_ball_mult_3}. The main challenge of the analysis is the trajectory of the ball according to the movement of the walls. In the following we show how the bumpers of the construction in Subsection~\ref{subsect:Time_Delay_of_the_Pinball} can be simulated by moving walls. Thus, we have that:

\begin{lemma}
	The $2$D Pinball Wizard problem with one-way gates, walls and moving walls
	simulates a pushdown automaton where the stack is represented by a time offset
	such that each step is represented by a constant number of reflections.
	\label{le:Pinball_noBumper}
\end{lemma}

\begin{proof}
	
As for our analyzes of the pinball problem with bumpers we assume that outside of the individual construction blocks, a ball always has a fixed speed of $v$. Hence the required time for the ball to reach the moving interval of the first moving wall from the start interval is $\frac{\delta_1}{v}$; and the 
required time of the ball to reach the end interval from moving interval of the second moving wall is $\frac{\delta_3}{v}$. Since both values are 
independent from the time offset $t_s$ we will ignore these parts in the following analysis.

%
%

We start with the multiplication of the time offset by $2$.
Let us assume the following values
where $t_s$ is the starting time of the ball relative to a base timeline (marked as start line of moving interval of the rightmost moving wall).
We can divide the movement of the ball into three parts. 

\paragraph*{Part 1}
\begin{figure}[htb]
	\begin{center}
		\scalebox{.25}{
			\input{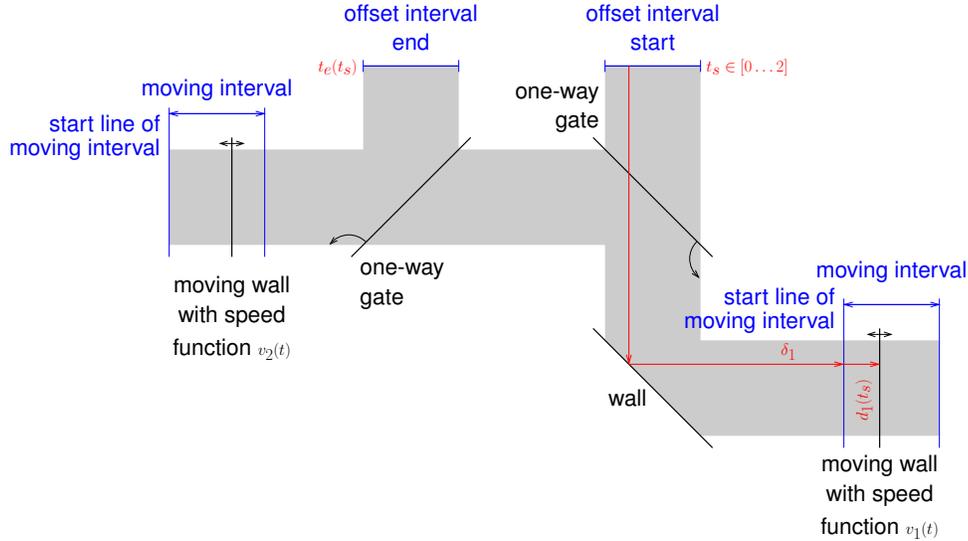}
		}
	\end{center}
	\caption{Multiplying the time offset by a factor using moving walls: Definition of $\delta_1$\label{fig:tool_time_ball_mult_1}}
\end{figure}
\FloatBarrier

For the first part the ball goes from the 
starting line of the first moving wall to the collision with the first moving wall.
\begin{itemize}
\item We assume that the length of the moving interval of the first moving wall is $1$.
\item $v$ is the initial speed of ball and its final speed.
\item The ball starts at the start line of the moving interval at time step $t_s$.
\item $a_1(t)$ is the acceleration of the first moving wall at time step $t$. We assume that
the acceleration starts at time $0$.
\item $t_{c,1}(t_s)$ is the time when the ball collides with the first moving wall.
\item $d_1(t_s)$ is the distance between the start line of the moving interval of the first moving wall and the position of the moving wall where the ball collides with the moving wall. Note that $d_1(t_s) = (t_{c,1}(t_s)-t_s)\cdot v$.
\item $v_1(t_s)=t_{c,1}(t_s)\cdot a_1(t_{c,1}(t_s))$ is the speed of the moving wall when the ball collides with the first moving wall.
\item $v_{b,1}(t_s)$ is the speed of the ball after the collision with the first moving wall.
Note that $v_{b,1}(t_s) = v - 2\cdot v_1(t_s)$.
\end{itemize}
Note that 
$$
t_{c,1}(0)  =  0,\qquad
d_1(0)  =  0, \qquad
v_1(0)  =  0,\qquad\text{and}\qquad
v_{b,1}(0)  =  v\ .
$$
The dependence between $a_1(t)$ and $t_{c,1}(t_s)$ follow from the following 
equation
\begin{eqnarray}
d_1(t_s) \ \ = \ \ (t_{c,1}(t_s)-t_s) \cdot v \ \ = \ \ \int_{x=0}^{t_{c,1}(t_s)} a_1(x)\cdot x\ d x\ 
\label{equ:pinball_without_bumpers_01}
\end{eqnarray}

\paragraph*{Part 2}
\begin{figure}[htb]
	\begin{center}
		\scalebox{.25}{
			\input{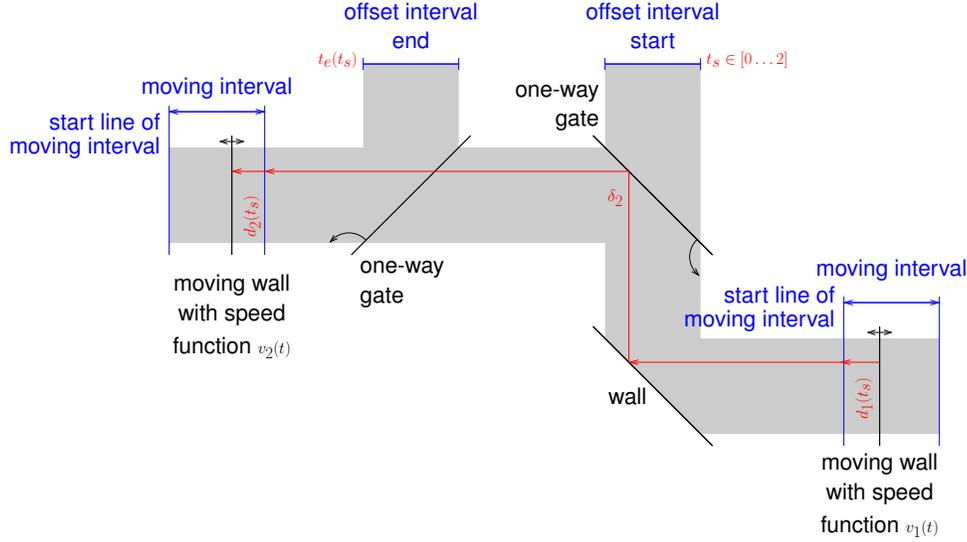}
		}
	\end{center}
	\caption{Multiplying the time offset by a factor using moving walls: Definition of $\delta_2$  \label{fig:tool_time_ball_mult_2}}
\end{figure}
\FloatBarrier

In the second part, the ball goes from the 
the first moving wall to the second moving wall.
\begin{itemize}
\item $\delta_2$ is the distance between the 
start line of the moving interval of the first moving wall and the 
end line of the moving interval of the second moving wall.
\item $1$ is the length of the moving interval of the second moving wall.
\item $v_{b,1}(t_s)$ is the speed of ball within this part.
\item The ball reaches the end line of the moving interval 
of the second moving wall at time step $t_h(t_s)$.
Note that $t_h(t_s) = t_{c,1}(t_s)+\frac{\delta_2+d_1(t_s)}{v_{b,1}(t_s)}$. 
And note that in its movement the second moving wall moves from the 
left to the right, therefore we use $t_h=t_h(0)+\frac{1}{v}$.
\item $a_2(t)$ is the acceleration of the second moving wall. We assume that
the acceleration starts at time $t_h$. For simplicity we denote by the parameter $t$ the time after the movement of the moving wall has started, i.e. the time difference to $t_h$.
\item $t_{c,2}(t_s)$ is the time difference between $t_h$ and the time where the ball collides with the second moving wall.
\item $v_2(t_s) = a_2(t_{c,2}(t_s))\cdot t_{c,2}(t_s)$ will be the speed of the second moving wall when the ball collides with this wall.
\item $d_2(t_s)$ is $1$ less than the distance between the start line of the moving interval of the second moving wall and the position of the moving wall where the ball collides with the moving wall. That is, $d_2(t_s) = 1-t_{c,2}(t_s)\cdot v_{b,1}(t_s)$.
\item $v_{b,2}(t_s)=v$ is the speed of the ball after the collision with the second moving wall.
For the speed after the collision with the second moving wall we assume that $v = v_{b,1}(t_s) + 2\cdot v_2(t_s)$.
\end{itemize}
Note that 
$$
t_h  =  \frac{\delta_2+1}{v},\qquad
t_{c,2}(0)  =  0,\qquad
d_2(0)  =  1,\qquad
v_{b,1}(0)  =  v_{b,2}(t) \ = \ v,
$$
and
$$
v_1(t_s)  =  a_1(t_{c,1}(t_s))\cdot t_{c,1}(t_s)  =  a_2(t_{c,2}(t_s))\cdot t_{c,2}(t_s)  =  v_2(t_s)\ .
$$
The dependence between $a_2(t)$ and $t_{c,2}(t_s)$ follow from the following 
equation
\begin{eqnarray}
d_2(t_s) & = & 
\left(
t_{c,2}(t_s)+t_h    
-\left(t_{c,1}(t_s)+\frac{\delta_2+d_1(t_s)}{v_{b,1}(t_s)}\right)
\right) \cdot v_{b,1}(t_s)\nonumber\\
& = & \left(t_{c,2}(t_s)  +\frac{\delta_2+1}{v}
-t_{c,1}(t_s)
\right) \cdot v_{b,1}(t_s)
-\delta_2-d_1(t_s) \nonumber \\ 
& = & 1-\int_{x=0}^{t_{c,2}(t_s)} a_2(x)\cdot x\ d x\ .
\label{equ:pinball_without_bumpers_02}
\end{eqnarray}

\paragraph*{Part 3}

\begin{figure}[htb]
	\begin{center}
		\scalebox{.25}{
			\input{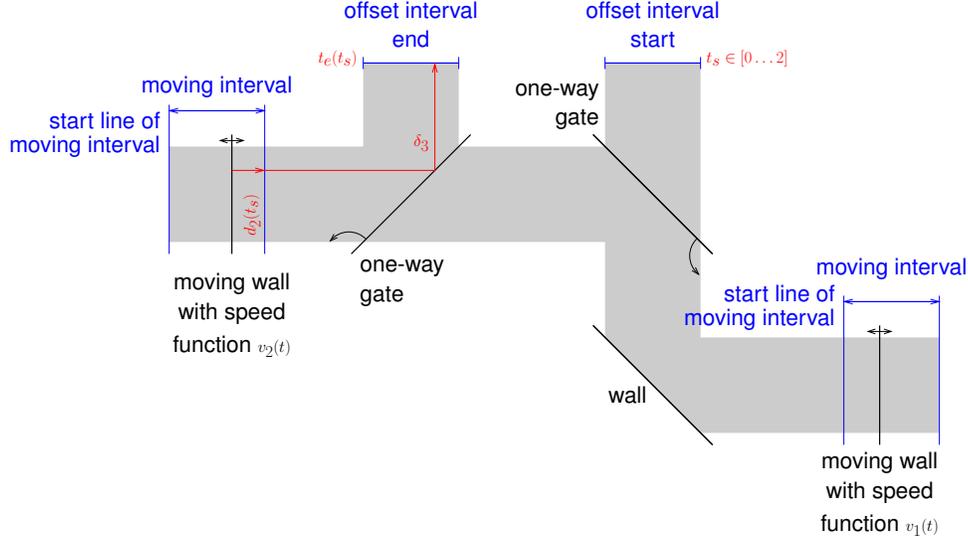}
		}
	\end{center}
	\caption{Multiplying the time offset by a factor using moving walls: Definition of $\delta_3$\label{fig:tool_time_ball_mult_3}}
\end{figure}
\FloatBarrier

For  the third part the ball goes from the 
the second moving wall to the end line.
\begin{itemize}
\item The ball reaches the end line of the moving interval of the second moving mirror for the last time at time step $t_3(t_s)$.
Note that $t_3(t_s) = t_h+t_{c,2}(t_s)+\frac{d_2(t_s)}{v}$. We use $t_e=t_3(0)$; this implies
$$
t_e  =  \frac{\delta_2+2}{v}\ .
$$
\end{itemize}
Our goal is now to determine functions for the acceleration $a_1(t_s)$ and $a_2(t_s)$ 
such that $t_3(t_s) = 2\cdot t_s + c$ for a value 
$$
c \ \ = \ \ t_e \ \ = \ \ \frac{\delta_2+2}{v}
$$ 
which does not depend in $t_s$.

Analysing from the time when the ball crosses the end line, the dependence between $a_2(t)$ and $t_{c,2}(t_s)$ follow from the following 
equation 
\begin{eqnarray}
d_2(t_s) & = & 
(t_{3}(t_s) - (t_{c,2}(t_s)+t_h)) \cdot v \ \ = \ \ 
(2\cdot t_s-t_{c,2}(t_s))\cdot v + 1 \nonumber\\
& = &
1-\int_{x=0}^{t_{c,2}(t_s)} a_2(x)\cdot x\ d x\ .
\label{equ:pinball_without_bumpers_03}
\end{eqnarray}
It might be useful to change the upper bounds of the Equations~\ref{equ:pinball_without_bumpers_01}, \ref{equ:pinball_without_bumpers_02}, and~\ref{equ:pinball_without_bumpers_03} to $t_s$. Since for a continuous differentiable function we have
$$
\int_{g(\ell)}^{g(u)} f(x) dx \ \ = \ \ \int_{\ell}^{u} f(g(x))\cdot g'(x) dx
$$
we get
\begin{eqnarray*}
(t_{c,1}(t_s)-t_s) \cdot v & = & \int_{x=0}^{t_s} a_1(t_{c,1}(t))\cdot t_{c,1}(t)\cdot t_{c,1}'(t)\ d t\\
\left(t_{c,2}(t_s)  +\frac{\delta_2+1}{v}
-t_{c,1}(t_s)
\right) \cdot v_{b,1}(t_s) \hspace*{1cm} & & \\
-\delta_2-d_1(t_s) & = & 1-\int_{x=0}^{t_s} a_2(t_{c,2}(t))\cdot t_{c,2}(t)\cdot t_{c,2}'(t)\ d t\\
(2\cdot t_s-t_{c,2}(t_s))\cdot v + 1 & = &
1-\int_{x=0}^{t_s} a_2(t_{c,2}(t))\cdot t_{c,2}(t)\cdot t_{c,2}'(t)\ d t\ .
\end{eqnarray*}
Transforming the equations into
differential equations lead to
\begin{eqnarray*}
v\cdot t_{c,1}'(t_s)- v & = & a_1(t_{c,1}(t_s))\cdot t_{c,1}(t_s)\cdot t_{c,1}'(t_s)\\
\left(t_{c,2}(t_s)  +\frac{\delta_2+1}{v}
-t_{c,1}(t_s)
\right) \cdot v_{b,1}'(t_s) \hspace*{1cm} & & \\
+ \left(t_{c,2}'(t_s) 
-t_{c,1}'(t_s)
\right) \cdot v_{b,1}(t_s)
-d_1'(t_s) & = & - a_2(t_{c,2}(t_s))\cdot t_{c,2}(t_s)\cdot t_{c,2}'(t_s)\\
2\cdot v-v\cdot t_{c,2}'(t_s) & = &
-a_2(t_{c,2}(t_s))\cdot t_{c,2}(t_s)\cdot t_{c,2}'(t_s)\ .
\end{eqnarray*}
or after the substitution with
$$
v_1(t_s) \ = \ a_1(t_{c,1}(t_s))\cdot t_{c,1}(t_s) \ = \ a_2(t_{c,2}(t_s))\cdot t_{c,2}(t_s) \ = \ v_2(t_s)
$$
and by the definition of $d_1(t_s)$ in the first part of this analysis, we get
%
%
\begin{eqnarray}
v\cdot (t_{c,1}'(t_s)- 1) & = & v_1(t_s)\cdot t_{c,1}'(t_s)
\label{equ:pinball_without_bumpers_03a}\\
\left(t_{c,2}(t_s)  +\frac{\delta_2+1}{v}
-t_{c,1}(t_s)
\right) v_{b,1}'(t_s) \hspace*{1cm} & & \nonumber \\
+\left(t_{c,2}'(t_s) 
-t_{c,1}'(t_s)
\right)  v_{b,1}(t_s)
& = & v_1(t_s) (t_{c,1}'(t_s)-t_{c,2}'(t_s))
\label{equ:pinball_without_bumpers_03b}\\
v\cdot (t_{c,2}'(t_s)-2) & = &
v_1(t_s)\cdot t_{c,2}'(t_s)\ .
\label{equ:pinball_without_bumpers_03c}
\end{eqnarray}
From Equation~\ref{equ:pinball_without_bumpers_03a} and~\ref{equ:pinball_without_bumpers_03c} we can conclude that
$$
\frac{t_{c,2}'(t_s)-2}{t_{c,2}'(t_s)}  = 
\frac{t_{c,1}'(t_s)- 1}{t_{c,1}'(t_s)}
$$
and thus
$$
t_{c,1}'(t_s) = 
\frac{t_{c,2}'(t_s)}{2}
$$
After integration we have
$$
t_{c,1}(t_s) = 
\frac{t_{c,2}(t_s)}{2}+ c_0
$$
with $c_0 = \frac{t_{c,2}(0)}{2}-t_{c,1}(0) = 0$, i.e.
\begin{eqnarray}
t_{c,2}(t_s)
& = & 2\cdot t_{c,1}(t_s)
\label{equ:pinball_without_bumpers_03d}
\end{eqnarray}
and since $a_1(t_{c,1}(t_s))\cdot t_{c,1}(t_s) = a_2(t_{c,2}(t_s))\cdot t_{c,2}(t_s)$
we have
\begin{eqnarray}
a_2(t) & = & \frac12\cdot a_1(t/2)\ .
\label{equ:pinball_without_bumpers_03e}
\end{eqnarray}

Combining the integral Equations~\ref{equ:pinball_without_bumpers_01}, \ref{equ:pinball_without_bumpers_02}, and~\ref{equ:pinball_without_bumpers_03} we get:
\begin{eqnarray*}
 \left(t_{c,2}(t_s)  +\frac{\delta_2+1}{v}
-t_{c,1}(t_s)
\right) \cdot v_{b,1}(t_s) \hspace*{1cm} & & \\
-\delta_2-(t_{c,1}(t_s)-t_s) \cdot v & = & 
(2\cdot t_s-t_{c,2}(t_s))\cdot v + 1 \\
 \left(t_{c,2}(t_s)  +\frac{\delta_2+1}{v}
-t_{c,1}(t_s)
\right) \cdot v_{b,1}(t_s) \hspace*{1cm} & & \\
-t_{c,1}(t_s) \cdot v & = & 
(t_s-t_{c,2}(t_s))\cdot v + 1 + \delta_2\\
 \left(t_{c,2}(t_s)  +\frac{\delta_2+1}{v}
-t_{c,1}(t_s)
\right) \cdot (v-2\cdot v_1(t_s))
& = & 
-(t_{c,2}(t_s)-t_{c,1}(t_s))\cdot v\\
& & \hspace*{1cm} + t_s\cdot v + 1 + \delta_2\\
 \left(t_{c,2}(t_s)  +\frac{\delta_2+1}{v}
-t_{c,1}(t_s)
\right) \cdot (v-2\cdot v_1(t_s))
& = & 
\left(t_{c,1}(t_s)-t_{c,2}(t_s)-\frac{\delta_2+1}{v}\right)\cdot v\\
& & \hspace*{1cm} + t_s\cdot v + 2\cdot(1 + \delta_2)\\
 2\cdot \left(t_{c,2}(t_s)  +\frac{\delta_2+1}{v}
-t_{c,1}(t_s)
\right) \cdot (v-v_1(t_s))
& = & t_s\cdot v + 2\cdot(1 + \delta_2)\\
 \left(t_{c,2}(t_s)  +\frac{\delta_2+1}{v}
-t_{c,1}(t_s)
\right) \cdot (v-v_1(t_s))
& = & \frac{t_s\cdot v}{2} + \delta_2 + 1\ .
\end{eqnarray*}
From Equation~\ref{equ:pinball_without_bumpers_03d} we can conclude
\begin{eqnarray}
 \left(t_{c,1}(t_s)  +\frac{\delta_2+1}{v}
\right) \cdot (v-v_1(t_s))
& = & \frac{t_s\cdot v}{2} + \delta_2 + 1\ .
\label{equ:pinball_without_bumpers_10}
\end{eqnarray}
and therefore
\begin{eqnarray}
v_1(t_s) & = & \frac{v^2}{2}\cdot\frac{2\cdot t_{c,1}(t_s)  
 - t_s}{v\cdot t_{c,1}(t_s)  +\delta_2+1}
\label{equ:pinball_without_bumpers_10}\\
a_1(t_s) & = & \frac{v^2}{2}\cdot\frac{2\cdot t_{c,1}(t_s)  
 - t_s}{v\cdot t_{c,1}^2(t_s)  +\delta_2\cdot t_{c,1}(t_s)+t_{c,1}(t_s)}
\label{equ:pinball_without_bumpers_11}
\end{eqnarray}
We assume that $v_1(t_s)>0$, $d_1(t_s), d_2(t_s)\ge 0$, $d_1(t_s), d_2(t_s)\le 1$, $t_s\in [0..2]$, and $t_{c,1}(t_c), t_{c,2}(t_c)$ are strictly monotonically increasing. 
$$
\begin{array}[c]{rclcrclcrcl}
t_{c,1}(t_s) & \ge & t_s & \hspace*{5mm} &
(t_{c,1}(t_s)-t_s)\cdot v & \le & 1 & \hspace*{5mm} &
2\cdot t_{c,1}(t_s) & > & t_s\\
t_{c,1}'(t_s) & > & 1 & & 
t_{c,2}'(t_s) & > & 2 & &
t_{c,2}(t_s) & \ge & 2\cdot t_s\\
t_{c,2}(t_s) & \le & 2\cdot t_s + \frac1v & &
2\cdot v_{1}(t_s) & < & v\ .
\end{array}
$$
Let us choose 
$$t_{c,1}(t_s) = t_s\cdot\left(1+\frac{1}{4\cdot v}\right)$$
which implies that
\begin{eqnarray*}
t_{c,2}(t_s) & = & 2\cdot t_s\cdot\left(1+\frac{1}{4\cdot v}\right)\\
\end{eqnarray*}
for simplicity, let
$ 
\Delta_v = 1+\frac{1}{4\cdot v}.
$
This leads to 
\begin{eqnarray*}
v_1(t_s) & = & \frac{v^2}{2}\cdot\frac{2\cdot t_s\cdot\Delta_v  
 - t_s}{v\cdot t_s\cdot\Delta_v  +\delta_2+1}\\
%
a_1(t_s) & = & \frac{v^2}{2}\cdot\frac{2\cdot\Delta_v
	 - 1}{v\cdot t_s\cdot\Delta_v^2  +\delta_2\cdot\Delta_v+\Delta_v}\\
%
a_2(t_s) & = & \frac12\cdot \frac{v^2}{2}\cdot\frac{2\cdot \Delta_v
	 - 1}{v\cdot t_s/2\cdot\Delta_v^2  +\delta_2\cdot\Delta_v+\Delta_v}
\end{eqnarray*}

%
%

\end{proof}

Note that at the output interval the order of the space offset is reversed. To correct the order we have to add the construction of 
Figure~\ref{fig:revOffset_appendix} at the output interval of the time offset multiplier.

One should highlight here that from Figure~\ref{fig:tool_parabola_mirror_02_rpt}, one can see that for every offset of the ball that might occur within each of the push $0$, push $1$ and pop constructions, the ball will always need the same amount time to go from the horizontal start line to the finish line. This is possible because an extra delay can be introduced on the ball's movement by an additional substructure such that its time of arrival at the end line is well-defined. Hence, the operations of the space offset stack is independent of the time offset stack.

After we have implemented a gadget for multiplying the time offset by a constant we have to show how we can split the time interval $[0\ldots 2)$ into the subinterval $[0\ldots 1)$ (representing the top element 0 on our time offset stack) and the subinterval $[1\ldots 2)$ (representing the top element $1$ on our time offset stack) after multiplying the time offset with $2$. For this we use the construction in Figure~\ref{fig:if_else_time_pb}. 
 If we always do the multiplication with $2^{-1}$ two times and do the same for the multiplication with $2$, we can guarantee that the ball never arrives in the intervals $[\frac12\ldots 1)$ and $[\frac32\ldots 2)$ at our separator after the double multiplication with $2$.
 This gives the moving wall some time to move out of the way of the ball or to come back, respectively. Let us now analyze the movement of the moving wall in more detail.

For our analysis we will assume that a ball with speed $v$ (this is the speed of the ball outside of our gadgets, i.e. outside of the second part of the way of the ball through our multipliers) needs $\kappa$ time units of our time offset to pass one space unit of our space offset, i.e. the ball needs $\kappa\cdot v$ to pass the space interval $[0\ldots 1)$ of our space offset stack.
Note that any ball starting at the leftmost position of the start interval reaches the moving wall $2\cdot\kappa$ time units later than the 
corresponding ball starting at the same time at the rightmost position of the start interval (recall that this interval has the length of $2\cdot \kappa$). Hence, the latest ball from the time interval $[0..\frac12)$ reaches the moving wall $\frac12+2\cdot\kappa$ time units later than the earliest ball from this time interval. 
On the other hand the earliest ball from the second
time offset interval $[1..\frac32)$ can reach the original position of the moving wall $1$ time unit later than the earliest ball of the 
time offset interval $[0..\frac12)$. 
Hence the time for moving the wall out of the way of the ball is reduced to $\frac12-2\cdot\kappa$. 
Let us assume that we can move 
such a wall with a fraction of $\varepsilon$ of the speed $v$ of the ball and assume that we move the wall over a distance of $3\cdot\kappa$.

Let us assume that the earliest time a ball can reach the moving wall is $t$, then
the movement function of the wall has the following form:
\begin{itemize}
	\item in the time interval $[0\ldots \frac12+2\cdot\kappa)$ the moving wall will not move,
	\item in the time interval $[\frac12+2\cdot\kappa\ldots 1)$ the moving wall will move down with a speed of $\varepsilon\cdot v$,
	\item in the time interval $[1\ldots \frac32+2\cdot\kappa)$ the moving wall will not move, and
	\item in the time interval $[\frac32+2\cdot\kappa\ldots 2)$ the moving wall will move up with a speed of $\varepsilon\cdot v$.
\end{itemize}
From this movement we can derive the following inequality
$$
\frac{3\cdot\kappa}{\varepsilon\cdot v} \ \ \le \ \ \frac{\left(\frac12 -2\cdot \kappa\right)}{v}\ .
$$
Hence we can use 
$$
\kappa \ \ < \ \ \frac14\qquad\text{and}\qquad 
\varepsilon \ \ \ge \ \ \frac{6\cdot\kappa}{1-4\cdot\kappa}
$$
for our construction. This leads to a useful speed of the moving wall --- even if we assume $v$ to be at the speed of light --- if we assume that $\kappa$ is sufficiently small.

Now, as stated in Section~\ref{space_offset_stack} already, from Lemmas~\ref{le:pinball_space_offset_stack},~\ref{th:Pinball_Bumper} and~\ref{le:Pinball_noBumper} we can conclude the correctness of Theorem~\ref{th:Pinball_noBumper}.
\section{Ray Particle Tracing Problem is Turing Complete}
\label{sec:RayParticleTracing}

Let us now assume that bumpers and moving walls do not change the speed of the ball, i.e. our ball has always the same speed $v$. This is for example the case if we assume that the ball has the speed of light, leading to the Ray Particle Tracing problem. Note that bumpers behave like usual walls if the speed of a ball cannot be changed. The only effect of a moving wall on the ball might be a change in the distance which the ball has to go.

Within this problem we usually use mirrors instead of walls however, for uniformity in the Appendix, we will stay with the word 'wall'.

Recall that our constructions in Appendix~\ref{space_stack_app} does not use any bumper or moving wall at all, hence we can derive directly all our observations for the implementation of a space offset stack within the settings of the Ray Particle Tracing problem. Recall that
to simulate a two-stack PDA, the horizontal space offset on the start position of the ball can be used to implement a first stack, while for a second stack, a time delay introduced on the ball could be used. According to our observation about Appendix~\ref{space_stack_app} we have already an implementation for the space offset stack. In the following we will present an implementation for the time offset stack.

Recall that if a ball is always traveling with maximum speed and collides with a moving mirror, its speed remains unchanged. Hence, we have to use the distance of the moving wall to a given starting line to manipulate the time offset of a ball. Furthermore, if a ball never changes its speed then potential speed changes do not take place and hence we do not have to undo such changes. Thus our construction for the multiplication of the time offset is simpler than the construction in Appendix~\ref{delay-stack-ball_app} (see Figure~\ref{fig:tool_time_mult_1_comp}).

\begin{figure}[htb]
	\centerline{
		\scalebox{.18}{
			\input{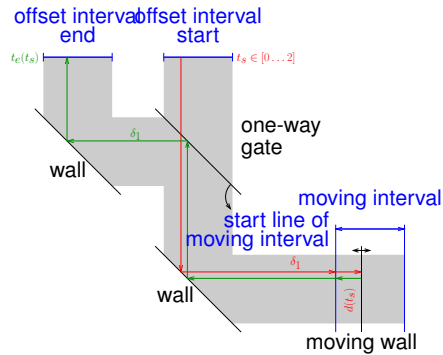}
		}
	}
	\caption{Construction for multiplying the time offset by a constant. 
		\label{fig:tool_time_mult_1_comp}}
\end{figure}
\FloatBarrier

For simplicity we will assume in the following that our ball travels over one space unit within one time unit. Note that such a space unit is not necessarily related to the space units of the space offset analyzed in the previous section.

\begin{figure}[htb]
	\begin{center}
		\scalebox{.25}{
			\input{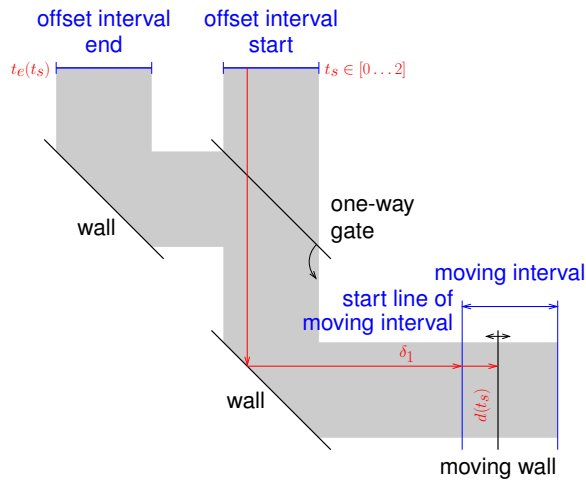}
		}
	\end{center}
	\caption{Construction for multiplying the time offset by a constant. The moving wall
		moves back and forth within the moving interval. Its speed is given by the
		function $d(t_s)$.
		The figure illustrates the
		motion of the ball from the horizontal starting line to the moving wall.
		If $t_s$ is the time offset of a ball (at the horizontal line "offset interval start") it will reach the starting line of the moving 
		wall at time $t_s+\delta_1$ and the moving wall at time $t_s+\delta_1+d(t_s)$.
		Since $\delta_1$ is a constant, we can assume that the position of the moving wall
		only depends on the time offset $t_s$.\label{fig:tool_time_mult_1}}
\end{figure}
\FloatBarrier

As presented in Figure~\ref{fig:tool_time_mult_1} and~\ref{fig:tool_time_mult_2} a ball with time offset 
$t_s$ will reach the horizontal finishing interval at time
$$
t_e\ =\ t_s+\delta_1+\delta_2+2\cdot d(t_s)
$$
independent from the space offset.
Let assume that $t_s\in [0\ldots 2)$. If the movement of the moving wall is chosen such that 
$$d(t_s) \ = \ \frac{t_s}{2}$$
we get
$$
t_e\ =\ 2\cdot t_s + \underbrace{\delta_1+\delta_2}_{\text{constant}}\ .
$$
And if we choose 
$$d(t_s) \ = \ 1-\frac{t_s}{4}$$
we get
$$
t_e\ =\ \frac{t_s}{2} + \underbrace{2+\delta_1+\delta_2}_{\text{constant}}\ .
$$
Note that the additive constants can be interpreted as a well-known constant movement of the base line of the 
time offset. Hence our construction gives the required multiplication of an time offset by $2$ and by $\frac{1}{2}$.

\begin{figure}[htb]
	\begin{center}
		\scalebox{.25}{
			\input{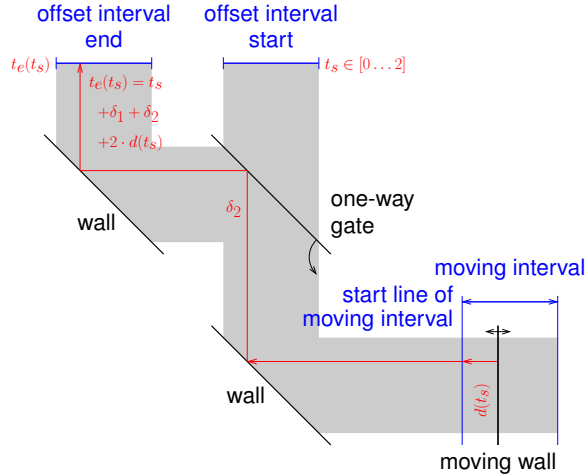}
		}
	\end{center}
	\caption{Construction for multiplying the time offset by a constant. The moving wall 
		is moving back and forth within the moving interval. Its speed is given by the
		function $d(t_s)$. The figure illustrates the
		motion of the ball from the moving wall to the horizontal finishing line. After 
		the reflection on the moving wall, the ball will again reach the
		starting line of the moving wall at time $t_s+\delta_1+2\cdot d(t_s)$ and it will
		reach the final horizontal line "offset interval end" at time $t_e=t_s+\delta_1+\delta_2+2\cdot d(t_s)$, where again $\delta_2$ is a constant only depending on the 
		size of the construction.\label{fig:tool_time_mult_2}
	}
\end{figure}
\FloatBarrier

Note that our construction for multiplying the time offset by $2$ or $\frac12$ reverses the ordering of the spatial offset again. But this can be rearranged as in the previous cases.

Recall that we assume that the time offset will be in the range $[0\ldots 2)$, but we left the unit of measurement open. Furthermore the speed of the moving walls depends on this unit. According to our assumption that the ball travels over one space unit within one time unit, $d(t_s)$ indicates that for the 
duplication of the time offset the wall has to move with a fraction of $\frac{1}{2}$ of the speed of the ball. If we assume that the ball is a ray particle traveling with the speed of light, this will be a huge value.

So let us assume that the maximum speed of a moving wall is at most a fraction of $\varepsilon\le \frac{1}{2}$ of the speed of the ball. If we iterate over the usage of the construction of Figure~\ref{fig:tool_time_mult_1} and~\ref{fig:tool_time_mult_2} then after 2 iterations, we get that 
$$
t_e\ =\ (t_s+2\cdot 2\cdot \varepsilon\cdot t_s) + \underbrace{2\cdot(\delta_1+\delta_2)}_{\text{constant}}
$$
and after $\xi$ iterations 
$$
t_e\ =\ (t_s+2\cdot \xi\cdot \varepsilon\cdot t_s) + \xi\cdot(\delta_1+\delta_2)\ .
$$
Hence if we choose 
$$
\xi \ \ = \ \ \left\lceil\frac{1}{2\cdot \varepsilon}\right\rceil
$$
and if we assume that the last moving wall within this sequence moves a little bit slower then we reach at the end
$$
t_e\ =\ 2\cdot t_s + \xi\cdot(\delta_1+\delta_2)\ .
$$
Note that if $\varepsilon$ is a constant also $\xi$ and therefore $\xi\cdot(\delta_1+\delta_2)$ are constants.

In a similar way we have to analyze the construction for the multiplication of the time offset with $\frac{1}{2}$. 
If we assume that $d(t_s)=1-\frac{t_s}{4}$, the moving wall will move with a fraction of $\frac{1}{4}$ of the speed of the ball. Again let us assume that
the maximum speed of a moving wall is at most a fraction of $\varepsilon\le \frac{1}{2}$ of the speed of the ball. If we iterate the usage of our construction 
then we get after 2 iteration that 
$$
t_e\ =\ (t_s-2\cdot 2\cdot \varepsilon\cdot t_s) + \underbrace{2\cdot(2+\delta_1+\delta_2)}_{\text{constant}}
$$
and after $\xi$ iterations 
$$
t_e\ =\ (t_s-2\cdot \xi\cdot \varepsilon\cdot t_s) + \xi\cdot(2+\delta_1+\delta_2)\ .
$$
Hence if we choose 
$$
\xi \ \ = \ \ \left\lceil\frac{1}{4\cdot \varepsilon}\right\rceil
$$
and if we assume that the last moving wall within this sequence moves a little bit slower then we reach at the end
$$
t_e\ =\ \frac{t_s}{2} + \xi\cdot(2+\delta_1+\delta_2)\ .
$$
Note that if $\varepsilon$ is a constant also $\xi$ and therefore $\xi\cdot(2+\delta_1+\delta_2)$ are constants.

For separating the two sub-intervals (one which represents the top entry 0 and the other the top entry 1 on the stack) we can again use our construction in Figure~\ref{fig:if_else_time_pb}. As we have discussed in Appendix~\ref{delay-stack-ball_app} we have to again double 
the number of multiplications by 2 or $2^{-1}$, respectively, to achieve a time gap for moving a separating moving wall out of the way of the ball.

Recall that our construction of a stack implementation by horizontal space offset does not influence the
time offset of a ray particle (if we ignore known constant delays). Analogously we can see that our 
construction of the stack implementation by the time offset does not influence the space offset.
Thus we have two independent stacks and therefore we can simulate a 1-tape DTM.

\section{Simulating a DTM}
\label{sec:SimDTM}

Within this part of the Appendix we will sketch how the Pinball Wizard problem and the Ray Particle Tracing problem can be used to simulate the computation of a 1-tape DTM (equivalently, a two-stack PDA). Hence both problems are Turing complete. Hence both problems are at least as hard as the Halting problem.

Assume a configuration of a 1-tape DTM (equivalently, a two-stack PDA). To represent the tape content left of the head we can use the time offset stack and to represent the tape content right of the head we can use the space offset stack. Given a 1-tape DTM with a binary input string $w$, let $S^1=\langle q_0, s^{1}_{0},w \rangle$ and $S^2=\langle s^{2}_{0}\rangle$ be the initial configuration of the space and time offset stacks, respectively where $q_0$ is the initial state of the DTM.
We assume that the states of the DTM are represented by binary numbers of length $\ell$. $s^{1}_{0}$ and $s^{2}_{0}$ are the bottom symbols encoded in the stacks.

At the start of every step in the simulation, if the current state is $q_i$ and $a_i$ is the symbol in the cell read by the DTM. Then we assume that $a_i$, all the elements to the right of $a_i$ and $q_i$ are in the space offset stack and the elements to the left of $a_i$ are in the time offset stack.
	
In each step, the current state $q_i$ and cell symbol $a_i$ are read from the space offset stack via a sequence of pop operations, each of which will result in a component which uniquely describe the starting point of the next step. Recall that each pop operation yields a binary decision tree over the possible output intervals regions, see Figure~\ref{fig:if_else_time_space}. In order to perform the transition function, we use a binary decision tree, where we branch out using the pop operation for every of the possibilities defined by the starting state and symbol read from the space offset stack.
Since this is a DTM, the next state $q_j$, the new cell symbol $a_j$, and the direction $\tau$ in which the head will move are uniquely defined. Next, we write $a_j$ to the current cell and move the tape head again by pop operations and push operations. Finally, we write the encoding of the new state $q_j$ to the space offset stack, combine the space offset intervals and time offset intervals (recall that the separators after the multiplications with 2 result in several alternative offset intervals) to a single unique output interval. Hence, the branches of the decision tree are combined again into a single space and time offset using one-way gates. 
The tool for combining intervals consisting of a wall and a one-way gate is shown in Figure~\ref{fig:tool_rejoin_intervals_01}. If necessary, additional delay on the particle can be introduced by strategically placed walls. Now we can start with the next simulation step.
	
For the moving walls we schedule the return time to the initial position according to the constant simulation time of each step of the Turing machine. So, we can reuse them for simulating in the next step.

\end{document}